%% file: main.tex
\newif\ifcomments   %
\newif\ifanon       %
\newif\ifcrypto     %
\newif\ifllncs      %

\ifllncs
  \documentclass[runningheads]{llncs}
\else
  \documentclass[11pt,pdfa]{article}
  \usepackage[in]{fullpage}
\fi

\usepackage{iftex}
\ifPDFTeX
  \usepackage[utf8]{inputenc}
  \usepackage[noTeX]{mmap}
  \usepackage[T1]{fontenc}
\fi
\ifLuaTeX
  \usepackage{luatex85}
  \usepackage[noTeX]{mmap}
\fi

\usepackage[shortlabels]{enumitem}
\usepackage{bm}
\usepackage{amsfonts}
\usepackage{amsmath}
\usepackage{amssymb}
\usepackage{amsthm}
\usepackage{xcolor}
\usepackage{graphicx}
\usepackage{qtree}
\usepackage{tree-dvips}
\usepackage{float}
\usepackage{hyperref}
\usepackage[nameinlink]{cleveref}
\usepackage{braket}
\usepackage{mathrsfs}
\usepackage{tikz}
\usepackage{qcircuit}

\ifllncs
  
\else
  \hypersetup{colorlinks={true},linkcolor={blue},citecolor=magenta}

  \theoremstyle{plain}
  \newtheorem{theorem}{Theorem}[section]
  \newtheorem{lemma}[theorem]{Lemma}
  \newtheorem{corollary}[theorem]{Corollary}
  \newtheorem{definition}[theorem]{Definition}
  \newtheorem{remark}[theorem]{Remark}

  \usepackage[style=alphabetic,minalphanames=3,maxalphanames=4,maxnames=99,backref=true]{biblatex}
  \renewcommand*{\labelalphaothers}{\textsuperscript{+}}
  \renewcommand*{\multicitedelim}{\addcomma\space}
  \newcommand{\email}[1]{\href{mailto:#1}{\texttt{#1}}}
\fi
\newtheorem{fact}[theorem]{Fact}

\ifcomments
  \newcommand{\todo}[1]{\textcolor{red}{TODO: #1}}
  \ifllncs
    \renewcommand{\note}[1]{\textcolor{blue}{NOTE: #1}}
  \else
    \newcommand{\note}[1]{\textcolor{blue}{NOTE: #1}}
  \fi
   \newcommand{\pnote}[1]{}
  \newcommand{\luowen}[1]{}
  \newcommand{\hnote}[1]{}  
\else
  \newcommand{\todo}[1]{}
  \ifllncs\renewcommand{\note}[1]{}\else\newcommand{\note}[1]{}\fi
  \newcommand{\pnote}[1]{}
  \newcommand{\luowen}[1]{}
  \newcommand{\hnote}[1]{}  
\fi

\input{headers}

\date{}
\title{Cryptography from Pseudorandom Quantum States}
\ifanon
  \newcommand{\citemanuscript}{\cite{Anonymous21}}
  \author{}
  \ifllncs
    \institute{}
  \fi
\else
  \newcommand{\citemanuscript}{\cite{AQY21}}
  \author{
  Prabhanjan Ananth\thanks{\email{prabhanjan@cs.ucsb.edu}}\\
  \small{\sl UCSB}
  \and Luowen Qian\thanks{\email{luowenq@bu.edu}}\\
   \small{\sl Boston University}
   \and Henry Yuen\thanks{\email{hyuen@cs.columbia.edu}}\\
   \small{\sl Columbia University}\vspace*{10pt}
  }
\fi

\begin{document}
\maketitle

\begin{abstract}
Pseudorandom states, introduced by Ji, Liu and Song (Crypto'18), are efficiently-computable quantum states that are computationally indistinguishable from Haar-random states. 
One-way functions imply the existence of pseudorandom states, but Kretschmer (TQC'20) recently constructed an oracle relative to which there are no one-way functions but pseudorandom states still exist. 
Motivated by this, we study the intriguing possibility of basing interesting cryptographic tasks on pseudorandom states.

We construct, assuming the existence of pseudorandom state generators that map a $\lambda$-bit seed to a $\omega(\log\secparam)$-qubit state, (a) statistically binding and computationally hiding commitments and (b) pseudo one-time encryption schemes. A consequence of (a) is that pseudorandom states are sufficient to construct maliciously secure multiparty computation protocols in the dishonest majority setting.

Our constructions are derived via a new notion called {\em pseudorandom function-like states} (PRFS), a generalization of pseudorandom states that parallels the classical notion of pseudorandom functions. Beyond the above two applications, we believe our notion can effectively replace pseudorandom functions in many other cryptographic applications.

\end{abstract}

\newpage

\tableofcontents

\newpage

\input{intro}

\ifcrypto\else
  \input{prelims}

\fi

\input{prs}

\input{prfsfromprs}
\input{encryption}

\input{commitment}

\section{Other Applications of PRFS}
In this section, we present some more applications of PRFS.
Note that the parameters required for these applications are beyond what we can construct from PRS in \Cref{sec:prfs-from-prs}.
\input{cpa}
\input{macs}
\input{garble}

\ifllncs
  \bibliographystyle{splncs04}
  \bibliography{tcs}
\else
  \printbibliography
\fi

\appendix

\section{Further details on instantiating MPC}
\label{sec:ot-details}
\input{ot-details}

\end{document}

%% file: headers.tex
\newcommand{\equad}{\mathrel{\phantom{=}}}

\newcommand{\hybrid}{\mathsf{H}}

\newcommand{\ketbra}[2]{\ket{#1}\!\bra{#2}}
\renewcommand{\cal}[1]{\mathcal{#1}}
\newcommand{\R}{\mathbb{R}}
\newcommand{\C}{\mathbb{C}}
\newcommand{\N}{\mathbb{N}}
\newcommand{\E}{\mathop{\mathbb{E}}}
\newcommand{\Tr}{\mathrm{Tr}}

\newcommand{\reg}[1]{\mathsf{#1}}

\newcommand{\eps}{\varepsilon}

\newcommand{\Haar}{\mathscr{H}}
\newcommand{\Test}{\mathsf{Test}}

\newcommand{\Enc}{\mathsf{Enc}}
\newcommand{\Dec}{\mathsf{Dec}}

\DeclareMathOperator{\TD}{TD}
\newcommand{\TraceDist}[2]{{\TD\left(#1, #2\right)}}
\newcommand{\norm}[1]{{\left\lVert#1\right\rVert}}

\newcommand{\operatornorm}[1]{{\norm{#1}}}

\newcommand{\bit}{{\{0, 1\}}}

\newcommand{\dmx}{\mathcal{D}}

\newcommand{\ip}[2]{\langle #1 | #2 \rangle}
\newcommand{\secparam}{\lambda}
\newcommand{\realexpt}{\mathsf{RealExpt}}
\newcommand{\idealexpt}{\mathsf{IdealExpt}}
\newcommand{\commit}{\mathsf{Commit}}
\newcommand{\reveal}{\mathsf{Reveal}}

\newcommand{\ciphertext}{{\bf c}}

\newcommand{\comlen}{{7\lambda}}

\DeclareFontFamily{U}{skulls}{}
\DeclareFontShape{U}{skulls}{m}{n}{ <-> skull }{}

%% file: intro.tex
\section{Introduction}
\label{sec:intro}

\noindent Assumptions are the bedrock of designing provably secure cryptographic constructions. Over the years, theoretical cryptographers have pondered over the precise assumptions needed to achieve cryptographic tasks, often losing sleep over this~\cite{Kil98}. The celebrated work of Goldreich~\cite{Gol90} shows that most interesting cryptographic tasks (encryption, commitments, pseudorandom generators, etc.) imply the existence of one-way functions, i.e., functions that can be efficiently computed in the forward direction but cannot be efficiently inverted. Thus it appears that the existence of one-way functions is a \emph{minimal} and \emph{necessary} assumption in cryptography.

Quantum information processing presents new opportunities for cryptography. Specifically, in many contexts the assumptions necessary for cryptographic tasks can be weakened with the help of quantum resources. To illustrate, the seminal work of Bennett and Brassard~\cite{BB84} showed that key exchange can be achieved unconditionally (i.e. without any computational assumptions) using quantum communication. In contrast, key exchange is known to require computational assumptions if the parties are restricted to classical communication. More recently, the work of Bartusek, Coladangelo, Khurana, and Ma~\cite{BartusekCKM21a} and that of Grilo, Lin, Song and Vaikuntanathan~\cite{GLSV21} demonstrate that quantum protocols for secure multiparty computation can be constructed from post-quantum one-way functions. On the other hand classical protocols for secure computation cannot be based (in a black-box fashion) on one-way functions alone~\cite{impagliazzo1989limits}.

These examples suggest that we revisit our belief about the necessity of certain cryptographic assumptions for quantum cryptographic tasks (tasks that make use of quantum computation and/or quantum communication). Specifically, it is not even clear whether one-way functions are even necessary in the quantum setting --- Goldreich's result~\cite{Gol90} only applies to classical cryptographic primitives and protocols.

Our work continues the research agenda carried out by our predecessors~\cite{Wei83,BB84,BBCS91,GLSV21,BartusekCKM21a}: {\em can we achieve cryptographic tasks using quantum communication in a world without one-way functions\footnote{Both the works~\cite{GLSV21,BartusekCKM21a} explicitly raised the question of basing secure computation on assumptions weaker than one-way functions.}?} %

\paragraph{Pseudorandom Quantum States.} 
Motivated by the question above, we turn to the notion of pseudorandom quantum states (abbreviated PRS) introduced by Ji, Liu and Song~\cite{ji2018pseudorandom}. A \emph{PRS generator} $G$ is a quantum polynomial-time (QPT) algorithm that, given input a $\secparam$-bit key, outputs an $n$-qubit quantum state with the guarantee that it is computationally indistinguishable from an $n$-qubit Haar random state (i.e. the uniform distribution over $n$-qubit pure states), even with many copies.
Ji, Liu and Song (and subsequently improved by Brakerski and Shmueli~\cite{BS19,BrakerskiS20}) show the existence of PRS assuming post-quantum one-way functions.

This notion is analogous to pseudorandom generators (PRGs) from classical cryptography which take as input a random seed of length $\secparam$, and deterministically outputs a larger string of length $n > \secparam$ that is computationally indistinguishable from a string sampled from the uniform distribution. Despite the analogy, it has not been obvious whether pseudorandom quantum states have much cryptographic utility outside of quantum money \cite{ji2018pseudorandom} (unlike PRGs, which are ubiquitous in cryptography).
Understanding the consequences of pseudorandom quantum states is particularly important in light of a recent result by Kretschmer~\cite{Kretschmer21}, who showed that there is a relativized world where $BQP = QMA$ (and thus post-quantum one-way functions do not exist) while pseudorandom states exist.
Kretschmer's result motivates us to focus the aforementioned research agenda on the following question: {\em what cryptographic tasks can be based solely on pseudorandom quantum states?}

\subsection{Our Results}
Our contributions in a nutshell are as follows:
\begin{itemize}
    \item We propose a new notion called \emph{pseudorandom function-like quantum states (PRFS)}.
    \item Using PRFS, we show how to build (a) statistically binding commitments and (b) pseudo one-time encryption schemes. As a consequence of (a), we obtain maliciously secure computation in the dishonest majority setting. 
    \item Finally, we show that for a certain range of parameters -- the same as what is needed for the above applications -- we can construct PRFS from a PRS. 
\end{itemize}
\noindent Before we present the definition of PRFS, we first highlight the need for defining a new notion by describing the challenges for constructing primitives directly from PRS.

\subsubsection{Challenges for Basing Primitives On PRS}
Although the closest classical analogue of a PRS generator is a PRG, the analogy breaks down in several critical ways. This makes it challenging to use PRS generators in the same way that PRGs are used throughout cryptography.
\par Specifically, PRS generators appear very \emph{rigid}, meaning that it seems challenging to take an existing PRS generator and generically increase or decrease its output length. Moreover, it is difficult to use output qubits of a PRS generator independently.

\paragraph{Inability to Stretch the Output.} A fundamental result about PRGs is that their \emph{stretch} (the output length as a function of the key length) can be amplified arbitrarily. In other words, given a PRG $G$ that maps $\secparam$ random bits to at least $(\secparam + 1)$ pseudorandom bits, one can construct a PRG with any polynomial output length.
This fact is implicitly used everywhere in cryptography; specifically, it gives us the flexibility to choose the appropriate stretch of PRG relevant for the application without having to worry about the underlying hardness assumptions.

If PRS generators are analogous to PRGs, then one would expect that a similar amplification result to hold: the existence of PRS with nontrivial output length would (hopefully) imply the existence of PRS with arbitrarily large output length. The natural approach to amplify the stretch of a PRG by iteratively composing it with itself does not immediately work with PRS for a number of reasons; for one, a PRS generator takes as input a classical key while its output is a quantum state! 

\paragraph{Inability to Shrink the Output.} To add insult to injury, it is not even obvious how to \emph{shrink} the output length of a PRS generator; this was also observed by Brakerski and Shmueli~\cite{BB21}. Classically, one can always discard bits from the output of a PRG, and the result is still obviously a PRG. However, discarding a single qubit of an $n$-qubit pseudorandom state will leave a mixed state that is easily distinguishable from an $(n-1)$-qubit Haar-random state.

\paragraph{Inability to Separate the Output.} Since the PRS output is highly entangled, it seems difficult to use the individual output qubits. As an example, suppose we want to encrypt a message of length $\ell$. In the classical setting, an $\ell$-bit output PRG can be used to encrypt a message of length $\ell$ by xor-ing the $i^{th}$ PRG output bit with the $i^{th}$ bit of the message. Implicitly, we are using the fact that the output of a PRG can be viewed a tensor product of bits and this feature of classical PRGs is mirrored by our notion of PRFS (explained next). On the other hand, if we have a single (entangled) PRS state (irrespective of the number of qubits it represents), it is unclear how to use each qubit to encode a bit; any operations performed on a single qubit could affect the other qubits that are entangled with this qubit.

\subsubsection{New Notion: Pseudorandom Function-Like States} Pseudorandom function-like states (abbreviated PRFS) is a generalization of PRS, where the same key $k$ can be used to generate many pseudorandom states. In more details, a \emph{$(d,n)$-PRFS generator} $G$ is a QPT algorithm that, given as input a key $k \in \bit^\secparam$ and an input $x \in \bit^d$, outputs a $n$-qubit quantum state $\ket{\psi_{k,x}}$,
satisfying the following pseudorandomness property: no efficient adversaries can distinguish between multiple copies of the output states $\left(\ket{\psi_{k,x_1}},\ldots,\ket{\psi_{k,x_s}} \right)$ from a collection of states $\left(\ket{\vartheta_1},\ldots,\ket{\vartheta_s}\right)$ where each $\ket{\vartheta_i}$ is sampled independently from the Haar distribution; furthermore, the indistinguishability holds even if the inputs $x_1,\ldots,x_s$ are chosen by the adversary.
This is formalized in \Cref{def:prfs}.

\paragraph{An Alternate Perspective: Tensor Product PRS generators.}
If PRS generators are analogous to classical pseudorandom generators, then PRFS generators are analogous to classical pseudorandom \emph{functions} (hence the name pseudorandom \emph{function-like}). A PRS generator outputs a single state per key $k$. On the other hand, we can think of PRFS as a {\em relaxed} notion of PRS generator that on input $k$ outputs a \emph{tensor product} of states $\ket{\psi_0} \otimes \ket{\psi_1} \otimes \cdots \otimes \ket{\psi_{2^{d}-1}}$ where each $\ket{\psi_i}$, is indistinguishable from a Haar-random state. 
\par The tensor product feature is quite useful in applications, as we will see shortly. 

\paragraph{Additional Observations.} Some additional observations of PRFS are in order: 
\begin{itemize} 

\item Assuming one-way functions, we can generically construct $(d,n)$-PRFS from any $n$-qubit PRS for any polynomial $d, n$. To compute PRFS on key $k$ and input $x$, first compute a classical PRF on $(k,x)$ and use the resulting output as a key for the $n$-qubit PRS. Since $n$-qubit PRS can be based on (post-quantum) one-way functions~\cite{ji2018pseudorandom,BrakerskiS20}, this shows that even PRFS can be based on (post-quantum) one-way functions. 
\item In the other direction, we can construct $n$-qubit PRS from any $(d,n)$-qubit PRFS. On input $k$, the PRS simply outputs the result of PRFS on input $(k,0)$. 

\item Another interesting aspect about PRFS is that it too, like PRS, is separated from (post-quantum) one-way functions. This can be obtained by a generalization of Kretschmer's result~\citemanuscript{}.

\end{itemize} 

\subsubsection{Implications}
We show that PRFS can effectively replace the usage of pseudorandom generators and pseudorandom functions in many primitives one learns about in ``Cryptography 101''. 
Specifically, we focus on two applications of PRFS generators. Later we will show that in fact that we can achieve these two applications from PRS generators only. 

\paragraph{Implication 1. One-time Encryption with Short Keys and Long Messages.}
As a starter illustration of the usefulness of PRFS, we construct from a PRFS generator $G$ a one-time encryption scheme for classical messages. The important feature of this construction is the fact that the message length is much larger than the key length. This is impossible to achieve information-theoretically, even in the quantum setting. This type of one-time encryption schemes, also referred to as {\em pseudo one-time pad}, is already quite useful, as it implies garbling schemes for P/poly~\cite{BMR90} and even garbling for quantum circuits~\cite{brakerski2020quantum}. %

\begin{theorem}[Informal; Pseudo One-time Pad]
\label{thm:intro:pqotp}
Assuming the existence of $(d,n)$-PRFS with\footnote{Recall that $\secparam$ is the key length.} $d=O(\log\secparam)$ and $n=\omega(\log\secparam)$,  there exists a one-time encryption scheme for messages of length $\ell=2^d$.
\end{theorem}

We emphasize that in the implication to one-time encryption, we only require PRFS with logarithmic-length inputs.

The construction is simple and a direct adaptation of the construction of one-time encryption from pseudorandom generators. To encrypt a message $x$ of length $\ell \gg \secparam$, output the state $G(k,(1,x_i)) \otimes \cdots \otimes G(k,(\ell,x_{\ell}))$, where $k \in \{0,1\}^{\secparam}$ is the symmetric key shared by the encryptor and the decryptor. The decryptor using the secret key $k$ can decode\footnote{In the technical sections, we define a QPT algorithm $\Test$ that given a state $\rho$ along with $k$, $x$, determines if $\rho$ is equal to the output $G(k,x)$. We show the existence of such a test algorithm for any PRFS.} the message $x$. The security of the encryption scheme follows from the pseudorandomness of PRFS.

\paragraph{Implication 2. Statistically binding commitment schemes.} We focus on designing commitment schemes with statistical binding and computational hiding properties. In the classical setting, this notion of commitment schemes can be constructed from any length-tripling PRG~\cite{Naor91}.
Recently, two independent works~\cite{GLSV21,BartusekCKM21a} showed that commitment schemes with aforementioned properties imply maliciously secure multiparty computation protocols with quantum communication in the dishonest majority setting.
Of particular interest is the work of~\cite{BartusekCKM21a} who construct the multiparty computation protocol using the commitment scheme as a \emph{black box}. In particular, their construction works even when the commitment scheme uses quantum communication. They then instantiate the underlying commitment scheme from post-quantum one-way functions.

We design commitment schemes based on PRFS instead of one-way functions. First, we present a new definition of the statistical binding property for commitment schemes that utilize quantum communication. The notion of binding for quantum commitment schemes is more subtle than that for classical commitment schemes and has been extensively studied in prior works~\cite{YWLQ15,Unruh16,FUYZ20,BartusekCKM21a,BB21}. Our definition generalizes all previously known definitions of statistical binding for quantum commitments, and suffice for applications such as secure multiparty computation. (Our definition is formally presented in \Cref{def:stat:binding}).

Then we show, assuming the existence of PRFS with certain parameters, the existence of quantum commitment schemes satisfying our definition.
\begin{theorem}[Informal]
\label{thm:intro:comm}
Assuming the existence of $(d,n)$-PRFS\footnote{To simplify the analysis, there is an additional technical property of the PRFS not mentioned here that is required by our construction, called \emph{recognizable abort} (\Cref{def:classicalgen}). All known constructions of PRFS and PRS (including ours) have the recognizable abort property.} where $2^{d} \cdot n \geq \comlen$, 
there exists a statistically binding and computationally hiding commitment scheme.
\end{theorem}

\noindent By plugging our commitment scheme into the framework of~\cite{BartusekCKM21a}, we obtain the following corollary.

\begin{corollary}[Informal]
\label{thm:intro:mpc}
Assumuing the existence of $(d,n)$-PRFS with $2^d \cdot n \ge \comlen$, there exists a maliciously secure multiparty computation protocol in the dishonest majority setting. 
\end{corollary}

\par Our construction is an adaptation of Naor's commitment scheme~\cite{Naor91}. We replace the use of the PRG in Naor's construction with a PRFS generator and the first message, which is a random string in Naor's construction, instead specifies a random Pauli operator. %

\paragraph{Other Implications.} Besides the above applications, we show that PRFS (with polynomially-long input length) can also be used to construct other fundamental primitives such as symmetric-key CPA-secure encryption (see \Cref{sec:cpa}) and message authentication codes (see \Cref{sec:macs}). Both primitives guarantee security in the setting when the secret key can be reused multiple times.

\par Unlike the previous applications (pseudo QOTP and commitments), the straightforward constructions of reusable encryption and MACs require PRFS generators with input lengths $\omega(\log \secparam)$ and $\ell$ respectively, where $\ell$ is the length of the message being authenticated. We do not know if such PRFS generators can be constructed from PRS generators in a black box way. Nonetheless, we believe these applications illustrate the usefulness of the concept of PRFS generators.

\subsubsection{Construction of PRFS}

Given the interesting implications of PRFS, the next natural step is to focus on constructing PRFS generators. We show that for some interesting range of parameters, we can achieve PRFS from any PRS. 
In particular, we show the following. 

\begin{theorem}[Informal]
\label{thm:intro:prfsfromprs}
For $d = O(\log\secparam)$ and $n = d + \omega(\log\log\secparam)$, assuming the existence of a $(d+n)$-qubit PRS generator, there exists a $(d,n)$-PRFS generator.
\end{theorem}

\noindent A surprising aspect about the above result is that the starting PRS's output length $d + n = \omega(\log\log\secparam)$ could even be much smaller than the key length $\secparam$. In contrast, classical pseudorandom generators with output length less than the input length are trivial.

We remark that if $d \ll \log\secparam$ then it is easy to build PRFS from PRS; chop up the key $k$ into $2^d$ blocks; to compute the PRFS generator with key $k$ and input $x$, compute the PRS generator on the $x^{th}$ block of the key. Unfortunately, PRFS with this range of parameters does not appear useful for applications because the seed length is too large. On the other hand, the construction of PRFS generators from PRS generators in \Cref{thm:intro:prfsfromprs} allows for $2^d$ to be an arbitrarily large polynomial in the key length. Note that this is sufficient for~\Cref{thm:intro:pqotp} and~\Cref{thm:intro:mpc}. We thus obtain the following corollary. 

\begin{corollary}
\label{cor:applications}
  Assuming $(2\log\lambda + \omega(\log\log\lambda))$-qubit PRS, there exist statistically binding commitment schemes and therefore secure computations.
  Assuming $\omega(\log\lambda)$-qubit PRS, there exist pseudo one-time pad schemes for messages of any polynomial length.
\end{corollary}

We remark that the assumptions of \Cref{cor:applications} on the PRS generators are essentially \emph{optimal}, in the sense that it is not possible to significantly weaken them. This is because commitment and pseudo one-time pad schemes require computational assumptions on the adversary; on other hand Brakerski and Shmueli~\cite{BrakerskiS20} demonstrate the existence a ``pseudo''-random state generator with output length $c \log \secparam$ for some constant $c < 1$ that is \emph{statistically secure}: in other words, the outputs of the generator are indistinguishable from Haar-random states by \emph{any} distinguisher (not just polynomial-time ones). 

Furthermore, it can be shown that PRS generators with $\log \secparam$-qubit outputs require computational assumptions on the adversary and that generators with $(1 + \eps) \log \secparam$-qubit outputs imply BQP $\neq$ PP~\citemanuscript{}.

\subsubsection*{Concurrent Work} 
A concurrent preprint of Morimae and Yamakawa~\cite{morimae2021quantum} also construct statistically binding and computationally hiding commitment schemes from PRS, adapting Naor's commitment scheme in a manner similar to ours. We note several differences between their work and ours, with regards to commitment schemes.
\begin{enumerate}
    \item They show a weaker notion of binding known as \emph{sum-binding}, which roughly says that the \emph{sum} of the probabilities that an adversarial committer can successfully decommit to the bit $0$ and the bit $1$ is at most a quantity negligibly close to $1$. This notion of binding is not known to be sufficient for general quantum commitment protocols to conclude that PRS implies protocols for secure computation\footnote{However, in an updated draft of~\cite{morimae2021quantum}, the authors sketch how, for a special form of quantum commitment schemes, sum-binding does imply our notion of statistical binding.}. However, our notion of statistical binding (\Cref{def:stat:binding}) is sufficient for leveraging the machinery of~\cite{BartusekCKM21a} to obtain quantum protocols for secure computation. Moreover, our definition of statistical binding implies the sum-binding definition\footnote{The sum of probabilities that an adversarial decommitter can decommit to 0 and to 1 in the ideal world of our definition (\Cref{def:stat:binding}) and therefore they sum up to at most negligibly larger than 1 in the real world by our statistical binding guarantee.}.
    \item For the same level of statistical binding security, that is $O(2^{-\lambda})$, they require the existence of a PRS that stretches $\secparam$ random bits to $3\secparam$ qubits of Haar-randomness (i.e., they require the PRS generator to have \emph{stretch}), whereas our result assumes the existence of a PRS that maps $\secparam$ bits to $2\log\secparam + \omega(\log\log \secparam)$ qubits. On the other hand, they require the pseudorandomness/indistinguishability of a single copy of PRS state versus Haar random, while we require the pseudorandomness to hold again multiple copies, especially when the output length is short.
    \item The state generation guarantee required from the underlying PRS is much stricter in their setting. In our work, we require the underlying PRS to only satisfy recognizable abort (\Cref{def:classicalgen}) whereas in their work, the underlying PRS needs to satisfy a guarantee that is even stronger 
    than perfect state generation (\Cref{def:perfstgen}).
    \item Their commitment scheme is non-interactive whereas our commitment scheme is a two-message scheme. Furthermore, our protocol has a classical opening message while theirs is quantum. However, these differences are rather minor since we can easily adapt our construction to satisfy these requirements, and vice versa.
\end{enumerate}
We also note that the notion of PRFS, its implications and its construction from PRS is unique to our work.

\subsection{Discussion: Why Explore a World Without One-Way Functions?}

Before getting into the technical overview we address a common question: \emph{``Sure, it is interesting that one can construct commitment schemes and pseudo one-time pad schemes without one-way functions, but will this still matter if someone proves that (post-quantum) one-way functions exist?''}

Our view is the following: 
it is \emph{not} our goal to avoid one-way functions because we don't believe that they exist\footnote{The majority of the authors of this paper believe one-way functions exist.}. The main motivation is to gain a \emph{deeper understanding} of fundamental cryptographic primitives such as encryption and commitment schemes. As mentioned previously, it has been understood for many decades that these primitives are inseparable from one-way functions (even in a black box way) in the classical setting. We view our results as revising this understanding in the quantum world: one-way functions are not necessary for these primitives. %

Another motivation comes from complexity theory. An oft-repeated storyline is that if P = NP, then one-way functions would not exist and thus most cryptography would be impossible; this scenario has been coined by Impagliazzo as \emph{Algorithmica} as one of his five ``complexity worlds''~\cite{impagliazzo1995personal}. While most people believe that P $\neq$ NP, it is nonetheless scientifically interesting to study the consequences of other complexity-theoretic outcomes. Our work adds a twist to the usual P = NP storyline: perhaps \emph{QAlgorithmica} -- Impagliazzo's \emph{Algorithmica} plus quantum information 
-- can potentially support both an algorithmic \emph{and} cryptographic paradise. 

Finally, we believe that studying the possibilities of basing cryptography solely on quantum assumptions is extremely useful for deepening our understanding of quantum information. By restricting ourselves to \emph{not} use one-way functions, we force ourselves to use the unique properties of quantum mechanics to the hilt. For example, our constructions of PRFS generators, pseudo one-time pad and commitment schemes ultimately required us to make use of properties of pseudorandom states such as concentration of measure over the Haar distribution.

\medskip

Another question that often arises is: \emph{``Is there a candidate construction of PRS generators that do not (obviously) involve one-way functions?''} While Kretschmer~\cite{Kretschmer21} showed an oracle separation between pseudorandom states and one-way functions, this is an artificial setting where the oracle is constructed by sampling a Haar-random unitary. 

We claim that \emph{random quantum circuits} form natural constructions of pseudorandom states: the generator $G$ interprets the key $k$ as a description of a quantum circuit on $n$ qubits, and $G$ outputs the state $k \ket{0^n}$ (i.e. executes the circuit with the all zeroes input). It has been conjectured in a number of settings that random quantum circuits have excellent pseudorandom properties. For example, the quantum supremacy experiments of Google~\cite{arute2019quantum} and UTSC~\cite{zhu2021quantum} are based on the premise that random $n$-qubit circuits of sufficiently large depth should generate states that are essentially Haar-random~\cite{harrow2018approximate}. Random quantum circuits have also been extensively studied as toy models of scrambling in black hole dynamics~\cite{brown2012scrambling,bouland2019computational,brandao2021models}.

It seems beyond the reach of present-day techniques to prove that polynomial-size random quantum circuits yield pseudorandom states; for one, doing so would separate BQP from PP~\cite{Kretschmer21}, which would be an incredible result in complexity theory. However, this is a plausible candidate PRS generator, and arguably this construction does not involve one-way functions at all.

\input{techoverview}

\subsection{Future Directions}

We end this section with some future directions and open questions.

\paragraph{Properties of pseudorandom states.} Given a PRS generator $G$ mapping $\lambda$-bit keys to $n$-qubit states, is it possible to construct in as black-box fashion as possible, a PRS generator $G'$ with longer output length (but same length key)? In other words, it is possible to arbitrarily \emph{stretch} the output of a PRS? 

Is it possible to construct PRFS generators (with polynomial-length inputs) from PRS generators in a black-box fashion? Are there separations?

\paragraph{More applications of pseudorandom states.} One of Impagliazzo's ``five worlds'' is called \emph{MiniCrypt}, which represents a world where one-way functions exist but we do not have public-key cryptography. In this world, applications such as symmetric-key encryption, commitment schemes, secure multiparty computation, and digital signatures are possible to achieve. %

It appears that we can obtain most MiniCrypt primitives from PR(F)S; for example this paper shows that we can get symmetric-key encryption, commitments, and secure multiparty computation. However it is a tantalizing open question of whether we can also build digital signatures from PR(F)S. Morimae and Yamakawa show that an analogue of one-time Lamport signatures can be constructed from PRS~\cite{morimae2021quantum}, but obtaining many-time signatures from PR(F)S seems more challenging.

More generally, what are other cryptographic applications of pseudorandom states?

\paragraph{Other quantum assumptions.} What are other interesting ``fully quantum'' assumptions that can we base cryptography on? Can we base cryptography on the assumption $\mathsf{BQP} \neq \mathsf{PP}$? We note that Chia, Chou, Zhang, Zhang also suggest the possibility of basing cryptography on the assumption that a quantum version of the Minimum Circuit Size Problem is hard~\cite[Open Problem 9]{chia2021quantum}.

\ifcrypto\else
\subsubsection{Organization}
We present the definitions of PRS and PRFS generators in \Cref{sec:prs}, as well as prove basic properties of them. We present a construction of PRFS generators from PRS generators in \Cref{sec:prfs-from-prs}. We present our quantum pseudo one-time pad scheme in \Cref{sec:qotp}. We present our quantum commitment scheme in \Cref{sec:commitment}. We present our many-time encryption scheme in \Cref{sec:cpa} and our message authentication code scheme in \Cref{sec:macs}.
\fi

\ifanon\else
\section*{Acknowledgements}
We thank Tomoyuki Morimae, Takashi Yamakawa, Jun Yan, and Fermi Ma for their very helpful feedback and discussions about pseudorandom quantum states.
HY is supported by AFOSR award FA9550-21-1-0040 and NSF CAREER award CCF-2144219.
LQ is supported by DARPA under Agreement No. HR00112020023.
\fi

%% file: techoverview.tex
\subsection{Technical Overview}
\noindent We first describe the techniques behind the construction of pseudorandom function-like states from pseudorandom quantum states. Then, we will give an overview of the result of statistical binding commitments from PRFS.

\subsubsection{PRFS from PRS}
To construct a $(d,n)$-PRFS, we start with an $(n+d)$-qubit PRS. For the purposes of the current discussion, we will assume that PRS has {\em perfect state generation}. That is, the output of PRS is a pure state.  

\paragraph{Main Insight: Post-Selection.} The construction proceeds as follows: on input key $k$ and $x \in \{0,1\}^d$, first generate a $(d+n)$-qubit PRS state by treating $k$ as the PRS seed. As the PRS satisfies perfect state generation, the output is a pure state and we can write the state as $\ket{\psi}= \sum_{x \in \{0,1\}^d} \alpha_x \ket{x} \otimes \ket{\psi_x}$, where $\ket{\psi_{x}}$ is a $n$-qubit state. Suppose we post-select (i.e., condition) on the first $d$ qubits being in the state $\ket{x}$, the remaining $n$ qubits will be in the state $\ket{\psi_x}$, which we define to be the output of the PRFS on input $(k,x)$.
\par In general, we do not know how to perform post-selection in polynomial-time \cite{Aar05}. However, if the event on which we are post-selecting has an inverse polynomial (where the polynomial is known ahead of time) probability of occurring, then we can efficiently perform post-selection. That is, we repeat the following process $2^d \secparam$ number of times: generate $\ket{\psi}$ by computing the PRS generator on $k$ and then measure the first $d$ qubits in the computational basis. If the first $d$ qubits is $x$ then output the residual state (which is $\ket{\psi_x}$), otherwise continue. If in none of the $2^d \secparam$ iterations, we obtained the first $d$ qubits to be $x$, we declare failure and output some fixed state $\ket{\bot}$.

We prove that the above PRFS generator satisfies pseudorandomness by making two observations. \\

\noindent {\em Observation 1: Output of PRFS is close to $\ket{\psi_x}$.} We need to argue that the probability that the PRFS generator outputs $\ket{\psi_x}$ is negligibly (in $\secparam$) close to 1. This boils down to showing that with probability negligibly close to 1, in one of the iterations, the measurement outcome will be $x$. Indeed if $|\alpha_x|^2$ is roughly $\frac{1}{2^d}$ then this statement is true. But it is a priori not clear how to argue this. 
\par Towards resolving this issue, let us first pretend that $\ket{\psi}$ was instead drawn from the Haar measure. In this case, we can rely upon L\'{e}vy's Lemma (\Cref{fact:levy}) to argue that $|\alpha_x|^2$ is indeed close to $\frac{1}{2^d}$, with overwhelming probability over the Haar measure. Thus, if $\ket{\psi}$ was drawn from the Haar measure, the probability that the PRFS generator outputs $\ket{\psi_x}$ is negligibly close to 1. 
\par Now, let us go back to the case when $\ket{\psi}$ was a PRS state. Since the PRFS generator is a quantum polynomial-time algorithm, it cannot distinguish whether $\ket{\psi}$ was generated by PRS or whether it was sampled from the Haar measure. This means that the probability that it outputs $\ket{\psi_x}$, when $\ket{\psi}$ was a PRS state, should also be negligibly close to 1.
\par While ideally we would have liked the PRFS to have perfect state generation, the above construction still satisfies a nice property that we call {\em recognizable abort}: the output of the PRFS is either a pure state or it is some known pure state $\ket{\bot}$. 

\par All is left is to show that the post-selected state $\ket{\psi_x}$ is Haar random when $\ket{\psi}$ is Haar random. \\

\noindent {\em Observation 2: Post-selected Haar random state is also Haar random.}
Haar random states satisfy a property called unitary invariance: applying any unitary on a Haar random state yields a Haar random state. Consider the following distribution ${\cal R}$ of unitaries: $R = \sum_{x \in \{0,1\}^n} \ket{x}\bra{x} \otimes R_x$, where $R_x$ is a Haar random unitary. Now, applying $R$, where $R \leftarrow {\cal R}$, on a Haar random state $\ket{\psi}=\sum_{x \in \{0,1\}^d} \ket{x} \otimes \ket{\psi_x}$ yields a Haar random state. \par Thus, the following two processes yield the same distribution: 
\begin{itemize}
    \item Process 1: Sample $\ket{\psi} = \sum_{x \in \{0,1\}^d} \ket{x} \otimes \ket{\psi_x}$ be a Haar random state. Output $\ket{\psi_x}$. 
    \item Process 2: Sample a Haar random state $\ket{\psi}= \sum_{x \in \{0,1\}^d} \ket{x} \otimes \ket{\psi_x}$. Output $R_x \ket{\psi_x}$. 
\end{itemize}
Notice that the output distribution of Process 2 is Haar random since $R_x$ is a Haar random unitary. From this we can conclude that even the output distribution of Process 1 is also Haar random. 

\paragraph{Test Procedure.} Classical pseudorandom generators satisfy a verifiability property that we often take for granted: given a value $y$ and a seed $k$, we can successfully check if $y$ is obtained as an evaluation of a seed $k$. This feature is implicitly used in many applications of pseudorandom generators. We would like to have a similar feature even for pseudorandom function-like states. More specifically, we would like the following to hold: given a state $\rho$, a PRFS key $k$ and an input $x$, check if $\rho$ is close to the output of PRFS on $(k,x)$.

\par Let us start with a simple case when the PRFS satisfies perfect state generation property and moreover, PRFS generator is a unitary $G$. We can express PRFS state generation as follows: on input a key $k$, input $x$ and ancillas $\ket{k} \otimes \ket{x} \otimes \ket{0}$, $G$ outputs $\ket{\psi_{k,x}} \otimes \ket{\phi}$. The state $\ket{\psi_x}$ is deisgnated to be the PRFS state corresponding to input $x$ and the state $\ket{\phi}$ is discarded as the garbage state. 
\par Suppose we need to test if a state $\rho$ is the output of PRFS on $k$ and $x$. The test procedure is defined as follows:
\begin{itemize}
    \item Compute $G(\ket{k} \otimes \ket{x} \otimes \ket{0})$, 
    \item Swap the register containing the PRFS state with $\rho$, 
    \item Apply $G^{\dagger}$ on the resulting state and,
    \item Measure the resulting state and output 1 if the outcome is $(k,x,0)$. Otherwise, output 0. 
\end{itemize}
\par Since unitaries preserve fidelity between the states, we can show that the following holds: the above test procedure outputs 1 with probability proportional to $F(\rho,\ketbra{\psi_{k,x}}{\psi_{k,x}})$. More precisely, the test procedure outputs 1 with probability $\text{Tr}(\ketbra{\psi_{k,x}}{\psi_{k,x}} \rho)$. 
\par The above test procedure can be suitably generalized if the PRFS satisfies the (weaker) state generation with recognizable abort property. If the PRFS generator is a quantum circuit then we designate $G$, in the above test procedure, to be a purification of this quantum circuit.

\subsubsection{Statistical Binding Commitments}
We show how to construct statistical binding quantum commitments from PRFS. 

\paragraph{Definition.} A statistical binding quantum commitment scheme consists of two interactive phases between a sender and a receiver: a commit phase and a reveal phase. In both the phases, the communication between the parties can be quantum. In the commit phase, the sender commits to a bit $b$. In the reveal phase, the committer reveals $b$ and the receiver either accepts or rejects. 
\par We require that any (even unbounded) sender cannot commit to bit $b$ in the commit phase and then successfully open to $1-b$ in the reveal phase. Formalizing this can be tricky in the setting where the communication channel is quantum. For example, consider the following attack: an adversarial sender can send a uniform superposition of commitments of 0 and 1 and then open to one of them in the reveal phase. Any definition we come up should handle this attack. 

\par We propose an extractor-based definition. Consider an adversarial sender $S^*$. Let us define the ideal experiment as follows: execute the commit phase with $S^*$. After the commit phase, apply an extractor on the receiver's state. The output of the extractor is a bit $b^*$ along with the collapse state $\sigma_R$. Execute the reveal phase; let $b$ be the bit opened to by $S^*$. Output ${\sf Fail}$ if $b \neq b^*$ and $R$ accepts. Otherwise, output $S^*$'s final state (after the execution of the Reveal phase) along with $R$'s decision, which is either the decommitted bit of the sender or it is $\bot$. Similarly, we can define real experiment as follows: We execute the commit phase and the reveal phase between $S^*$ and $R$ and then output the final state of $S^*$ along with $R$'s decision. 
\par Going back to the earlier superposition attack, the extractor would, with equal probability, collapse to either commitment of 0 or collapse to commitment of 1. 
\par We say that the quantum commitment scheme satisfies statistical binding if the output distributions of the real and ideal experiments are statistically close. Our definition of statistical binding generalizes all the previous definitions of statistical binding in the context of quantum commitments~\cite{YWLQ15,Unruh16,FUYZ20,BartusekCKM21a,BB21}. Refer to~\Cref{sec:comm:def} for a detailed comparison with prior definitions.   
\par We also require the quantum commitment scheme to satisfy computational hiding: in the commit phase, any quantum poly-time receiver cannot tell apart whether the sender committed to 0 or 1. 

\paragraph{Construction.} Our construction is a direct adaptation of Naor's commitment scheme~\cite{Naor91}, i.e. the same protocol but simply replacing PRG with PRFS. We start with a $(d,n)$-PRFS, where $d=O(\log(\secparam))$ and $n \geq 1$. %
 
\begin{itemize}
    \item In the commit phase, the receiver sends a random $2^d n$-qubit Pauli $P=P_0 \otimes \cdots P_{2^d - 1}$ to the sender, where each $P_i$ is a $n$-qubit Pauli. The sender on input bit $b$, samples a key $k$ uniformly at random from $\{0,1\}^{\secparam}$. The sender then sends the state ${\bf c} = \bigotimes_{x \in [2^d]} P_i^b \sigma_{k,x} P_i^b$, where $\sigma_{k,x} = PRFS(k,i)$ to the receiver. 
    \item In the reveal phase, the sender sends $(k,b)$ to the receiver. The receiver accepts if $P^b {\bf c} P^b$ is a tensor product of the PRFS evaluations of $(k,x)$, for all $x=0,\ldots,2^d - 1$. 
\end{itemize}

\noindent From the pseudorandomness property of PRFS, hiding  follows. To prove that the above scheme satisfies binding, we describe the  extractor first. It again helps to think of PRFS as satisfying the perfect state generation property. The extractor applies the following projection $\{ \Pi_0,\Pi_1,I - (\Pi_0 + \Pi_1) \}$, where $\Pi_b$ projects onto the subspace spanned by $T_b=\left\{ \bigotimes_{x \in \{0,1\}^{2^{d}}} P^b \ketbra{\psi_{k,x}}{\psi_{k,x}} P^b\ :\ \forall k \in \{0,1\}^{\secparam} \right\}$, where $\ket{\psi_{k, x}}$ is the output of $PRFS(k,x)$. If $\Pi_b$ succeeds then $b$ is designated to be the extracted bit. At the core of proving the indistinguishability of the real and the ideal world is the following fact: applying a projector that projects onto $T_b$ (as done by the extractor), followed by the projector $\bigotimes_{x \in \{0,1\}^{2^{d}}} P^b \ketbra{\psi_{k,x}}{\psi_{k,x}} P^b$ (as done by the receiver) is the equivalent to only applying the projector $\bigotimes_{x \in \{0,1\}^{2^{d}}} P^b \ketbra{\psi_{k,x}}{\psi_{k,x}} P^b$. 
\par While our actual proof is  conceptually similar to the proof sketched above, there are some crucial details that we shoved under the rug. Firstly, $I - (\Pi_0 + \Pi_1)$ is not necessarily a projection since the projections $\Pi_0$ and $\Pi_1$ need not be orthogonal. Secondly, the PRFS generation is not perfect and we need to work with recognizable abort property. Nonetheless we circumvent these issues and show that the above construction still works. We refer the reader to Section~\ref{sec:comm:cons} for more details.

%% file: prelims.tex
\section{Preliminaries}
\label{sec:prelims}

We refer the reader to~\cite{nielsen_chuang_2010} for a comprehensive reference on the basics of quantum information and quantum computation. We use $I$ to denote the identity operator. We use $\dmx(\cal{H})$ to denote the set of density matrices on a Hilbert space $\cal{H}$. Let $\rho,\sigma \in \dmx(\cal{H})$ be density matrices. We write $\TD(\rho,\sigma)$ to denote the trace distance between them, i.e.,
\[
    \TD(\rho,\sigma) = \frac{1}{2} \| \rho - \sigma \|_1
\]
where $\norm{X}_1 = \Tr(\sqrt{X^\dagger X})$ denotes the trace norm.
We denote $\norm{X} := \sup_{\ket\psi}\{\braket{\psi|X|\psi}\}$ to be the operator norm where the supremum is taken over all unit vectors.
For a vector $x$, we denote its Euclidean norm to be $\norm{x}_2$.

\paragraph{General Measurements.} A \emph{general measurement} on a Hilbert space $\cal H$ is a set $M = \{ M_a \}_{a \in A}$ of operators acting on $\cal H$ indexed by some finite set $A$ of outcomes satisfying the completeness relation
\[
    \sum_{a \in A} M_a^\dagger M_a = I~.
\]
Applying the measurement $M$ to a density matrix $\rho \in \dmx(\cal H)$ corresponds to the following operation: outcome $a$ is obtained with probability $\Tr(M_a^\dagger M_a \rho)$, and the post-measurement state is defined to
\[
    \rho \mapsto \frac{ M_a \rho M_a^\dagger }{ \Tr(M_a^\dagger M_a \rho)}~.
\]

\paragraph{Haar Measure.} The Haar measure over $\C^d$, denoted by $\Haar(\C^d)$ is the uniform measure over all $d$-dimensional unit vectors. One useful property of the Haar measure is that for all $d$-dimensional unitary matrices $U$, if a random vector $\ket{\psi}$ is distributed according to the Haar measure $\Haar(\C^d)$, then the state $U\ket{\psi}$ is also distributed according to the Haar measure. For notational convenience we write $\Haar_m$ to denote the Haar measure over $m$-qubit space, or $\Haar((\C^2)^{\otimes m})$.

\begin{fact}
\label{fact:haar-avg}
We have
\[
    \E_{\ket{\psi} \leftarrow \Haar(\C^d)} \ketbra{\psi}{\psi} = \frac{I}{d}~.
\]
\end{fact}

The following result, known as \emph{L\'{e}vy's Lemma}, expresses strong concentration of measure for the Haar measure. 

\begin{fact}[L\'{e}vy's Lemma~\cite{hayden2006aspects}]
\label{fact:levy}
Let $f:\C^d \to \R$ be a function such that for all unit vectors $\ket{\psi},\ket{\phi}$ we have
\[
    | f(\ket{\psi}) - f(\ket{\phi}) | \leq K \cdot \norm{\ket{\psi} - \ket{\phi}}_2
\]
for some number $K > 0$. Then there exists a universal constant $C > 0$ such that
\[
    \Pr_{\ket{\psi} \leftarrow \Haar(\C^d)} \left [ | f(\ket{\psi}) - \E f | \geq \delta \right ] \leq \exp \Big ( - \frac{C d \delta^2}{K^2} \Big)
\]
where $\E f$ denotes the average of $f$ over the Haar distribution $\Haar(\C^d)$. 
\end{fact}

\subsection{Quantum Algorithms}
\label{sec:algorithms}

A quantum algorithm $A$ is a family of generalized quantum circuits $\{A_\lambda\}_{\lambda \in \N}$ over a discrete universal gate set (such as $\{ CNOT, H, T \}$). By generalized, we mean that such circuits can have a subset of input qubits that are designated to be initialized in the zero state, and a subset of output qubits that are designated to be traced out at the end of the computation. Thus a generalized quantum circuit $A_\lambda$ corresponds to a \emph{quantum channel}, which is a is a completely positive trace-preserving (CPTP) map. When we write $A_\lambda(\rho)$ for some density matrix $\rho$, we mean the output of the generalized circuit $A_\lambda$ on input $\rho$. If we only take the quantum gates of $A_\lambda$ and ignore the subset of input/output qubits that are initialized to zeroes/traced out, then we get the \emph{unitary part} of $A_\lambda$, which corresponds to a unitary operator which we denote by $\hat{A}_\lambda$. The \emph{size} of a generalized quantum circuit is the number of gates in it, plus the number of input and output qubits.

We say that $A = \{A_\lambda\}_\lambda$ is a \emph{quantum polynomial-time (QPT) algorithm} if there exists a polynomial $p$ such that the size of each circuit $A_\lambda$ is at most $p(\lambda)$. We furthermore say that $A$ is \emph{uniform} if there exists a deterministic polynomial-time Turing machine $M$ that on input $1^n$ outputs the description of $A_\lambda$. 

We also define the notion of a \emph{non-uniform} QPT algorithm $A$ that consists of a family $\{(A_\lambda,\rho_\lambda) \}_\lambda$ where $\{A_\lambda\}_\lambda$ is a polynomial-size family of circuits (not necessarily uniformly generated), and for each $\lambda$ there is additionally a subset of input qubits of $A_\lambda$ that are designated to be initialized with the density matrix $\rho_\lambda$ of polynomial length. This is intended to model nonuniform quantum adversaries who may receive quantum states as advice.
Nevertheless, the reductions we show in this work are all uniform.

The notation we use to describe the inputs/outputs of quantum algorithms will largely mimick what is used in the classical cryptography literature. For example, for a state generator algorithm $G$, we write $G_\lambda(k)$ to denote running the generalized quantum circuit $G_\lambda$ on input $\ketbra{k}{k}$, which outputs a state $\rho_k$. 

Ultimately, all inputs to a quantum circuit are density matrices. However, we mix-and-match between classical, pure state, and density matrix notation; for example, we may write $A_\lambda(k,\ket{\theta},\rho)$ to denote running the circuit $A_\lambda$ on input $\ketbra{k}{k} \otimes \ketbra{\theta}{\theta} \otimes \rho$. In general, we will not explain all the input and output sizes of every quantum circuit in excruciating detail; we will implicitly assume that a quantum circuit in question has the appropriate number of input and output qubits as required by context.

%% file: prs.tex
\section{Pseudorandom States}
\label{sec:prs}

The notion of pseudorandom states were first introduced by Ji, Liu, and Song in~\cite{ji2018pseudorandom}. We reproduce their definition here:

\begin{definition}[PRS Generator~\cite{ji2018pseudorandom}]
\label{def:vanilla-prs}
We say that a QPT algorithm $G$ is a \emph{pseudorandom state (PRS) generator} if the following holds. 
\begin{enumerate}
    
    \item \textbf{State Generation}. There is a negligible function $\eps(\cdot)$ such that for all $\lambda$ and for all $k \in \{0,1\}^\lambda$, the algorithm $G$ behaves as
    \[
        G_\lambda(k) = \ketbra{\psi_{k}}{\psi_{k}}.
    \]
    for some $n(\lambda)$-qubit pure state $\ket{\psi_k}$.

    \item \textbf{Pseudorandomness}. For all polynomials $t(\cdot)$ and QPT (nonuniform) distinguisher $A$ there exists a negligible function $\eps(\lambda)$ such that for all $\lambda$,  we have
    \[
        \left | \Pr_{k \leftarrow \{0,1\}^\lambda} \left [ A_\lambda (G_\lambda(k)^{\otimes t(\lambda)}) = 1 \right] - \Pr_{\ket{\vartheta} \leftarrow \Haar_{n(\lambda)}} \left [ A_\lambda (\ket{\vartheta}^{\otimes t(\lambda)}) = 1 \right] \right | \leq \eps(\lambda)~.
    \]
\end{enumerate}
We also say that $G$ is a $n(\lambda)$-PRS generator to succinctly indicate that the output length of $G$ is $n(\lambda)$.
\end{definition}

Ji, Liu, and Song showed that post-quantum one-way functions can be used to construct PRS generators. 

\begin{theorem}[\cite{ji2018pseudorandom,BrakerskiS20}]
  If post-quantum one-way functions exist, then there exist PRS generators for all polynomial output lengths.
\end{theorem}

\subsection{Pseudorandom Function-Like State (PRFS) Generators}

In this section, we present our definition of pseudorandom function-like state (PRFS) generators. PRFS generators generalize PRS generators in two ways: first, in addition to the secret key $k$, the PRFS generator additionally takes in a (classical) input $x$. The security guarantee of a PRFS implies that, even if $x$ is adversarily chosen, %
the output state of the generator is indistinguishable from Haar-random. The second way in which this definition generalizes the definition of PRS generators is that the output of the generator need not be a pure state.

\begin{definition}[PRFS generator]
\label{def:prfs}
We say that a QPT algorithm $G$ is a (selectively secure) \emph{pseudorandom function-like state (PRFS) generator} if for all polynomials $s(\cdot), t(\cdot)$, QPT (nonuniform) distinguishers $A$ and a family of indices $(\{x_1,\ldots, x_{s(\lambda)}\} \subseteq \{0,1\}^{d(\lambda)})_\lambda$, there exists a negligible function $\eps(\cdot)$ such that for all $\lambda$,
    \begin{align*}
        &\Big | \Pr_{k \leftarrow \{0,1\}^\lambda} \left [ A_\lambda( x_1,\ldots,x_{s(\lambda)},G_\lambda(k,x_1)^{\otimes t(\lambda)},\ldots, G_\lambda(k,x_{s(\lambda)})^{\otimes t(\lambda)}) = 1 \right] \\
        & \qquad \qquad - \Pr_{\ket{\vartheta_1}, \ldots,\ket{\vartheta_{s(\lambda)}} \leftarrow \Haar_{n(\lambda)}} \left [ A_\lambda( x_1,\ldots,x_{s(\lambda)}, \ket{\vartheta_1}^{\otimes t(\lambda)},\ldots, \ket{\vartheta_{s(\lambda)}}^{\otimes t(\lambda)}) = 1 \right] \Big | \leq \eps(\lambda)~.
    \end{align*}
We also say that $G$ is a $(d(\lambda),n(\lambda))$-PRFS generator to succinctly indicate that its input length is $d(\lambda)$ and its output length is $n(\lambda)$.
\end{definition}

Our notion of security here can be seen as a version of \emph{(classical) selective security}, where the queries to the PRFS generator are fixed before the key is sampled.  One could consider stronger notions of security where the indistinguishability property holds even when the adversary is allowed to query the PRFS generator adaptively, or even in superposition. We explore these stronger notions in forthcoming work~\citemanuscript{}.

\paragraph{State Generation Guarantees.} As mentioned above, our definition of PRFS generator does not require that the output of the generator is always a pure state. 
However, we will see later that a consequence of the PRFS security guarantee is that the output of the generator is \emph{close} to a pure state for an overwhelming fraction of keys $k$ (see \Cref{lem:prs-orthog}).

Nevertheless, for applications it is sometimes more useful to also consider a stronger guarantee on the state generation of a PRFS generator.

\begin{definition}[Perfect state generation]
\label{def:perfstgen}
  A $(d(\lambda),n(\lambda))$-PRFS generator $G$ satisfies perfect state generation, if for every $k \in \bit^\lambda$ and $x \in \bit^{d(\lambda)}$, there exists an $n(\lambda)$-qubit pure state $\ket\psi$ such that $G_\secparam(k, x) = \ketbra\psi\psi$.
\end{definition}

We observe that an $n(\lambda)$-PRS generator defined in \Cref{def:vanilla-prs} is by definition equivalent to an $(0, n(\lambda))$-PRFS generator with perfect state generation.

In general, it may be difficult to construct PR(F)S with perfect state generation as the state generation could occasionally fail; for example, the generator may perform a (quantum) rejection sampling procedure in order to output the state. The scalable PRS generators of Brakerski and Shmueli~\cite{BrakerskiS20} is an example of this. To capture a very natural class of PRFS generators (including the one constructed in this paper), we define the notion of a PRFS generator where $G(k,x)$ outputs a convex combination of a fixed pure state $\ket{\psi_{k,x}}$ or a known abort state $\ket{\bot}$.

\begin{definition}[Recognizable abort]
\label{def:classicalgen}
 A $(d(\secparam),n(\secparam))$-PRFS generator $G$ has the \emph{recognizable abort property} if for
  every $k \in \bit^\lambda$ and $x \in \bit^{d(\secparam)}$ there exists an $n(\secparam)$-qubit pure state $\ket\psi$ and $0 \le \eta \le 1$ such that $G_\secparam(k, x) = \eta\ketbra\psi\psi + (1 - \eta)\ketbra\bot\bot$, where $\bot$ is a special symbol\footnote{One can think of $\ket\bot$ as the $(n + 1)$-qubit state $\ket{100\cdots0}$ with the first qubit indicating whether the generator aborted or not. If the generator doesn't abort, then it outputs $\ket{0} \otimes \ket{\psi}$ for some pure state $\ket{\psi}$ (called the \emph{correct output state} of $G$ on input $(k,x)$). The distinguisher in the definition of PRFS generator would then only get the last $n$ qubits as input.
  }.
\end{definition}
Note that this definition alone does not have any constraint on $\eta$ being close to $1$. However, the security guarantee of a PRFS generator implies that $\eta$ will be negligibly close to 1 with overwhelming probability over the choice of $k$.\footnote{The argument is as follows: if $\eta$ were on average noticeably far from $1$, then a purity test using SWAP tests would distinguish the outputs from Haar random states which are pure. This is formalized in \Cref{lem:prs-orthog}.}
We also note that a PRFS generator with perfect state generation trivially has the recognizable abort property with $\eta = 1$ for all $k, x$.

\ifcrypto
\else
  \input{prs-properties}
\fi

\subsection{Testing Pseudorandom States}

Given a state $\rho$, it is useful to know whether it is the output of a PRFS generator with key $k$ and input $x$. 
One approach would be to invoke the generator to get some number of copies and perform SWAP tests. Unfortunately, this approach would only achieve polynomially small error,
which is undesirable for cryptographic applications where we want negligible security.
Another approach is to ``uncompute'' the state generation.
The issue with this approach is that it is not clear how to do it when the state generation is not perfect, or if it outputs some additional auxiliary states that we do not know how to uncompute.

In the following, we will show how to use the generator in a semi-black-box way to test any PRFS states.
We first state a general Lemma that shows how to convert any circuit that generates a state $\rho$ into a tester (of sorts) for the state $\rho$. 
\begin{lemma}[Circuit output tester]
  \label{prop:test-channel}
  Let $G$ denote a (generalized) quantum circuit that takes no input and outputs an $n$-qubit mixed state $\rho$.
  Then there exists a circuit $\Test$ with boolean output such that:
  \begin{enumerate}
      \item For all density matrices $\sigma_{\reg{E} \reg{Q}}$ where $\reg{Q}$ is an $n$-qubit register, applying the circuit $\Test$ on register $\reg{Q}$ yields the following state on registers $\reg{E} \reg{F}$ where $\reg{F}$ stores the decision bit:
      \[
        (I_{\reg{E}} \otimes \Test_{\reg{Q}}) (\sigma_{\reg{E} \reg{Q}}) = \sum_b \Tr_{\reg{Q}} \Big( (I_\reg{E} \otimes M_b) \sigma_{\reg{E} \reg{Q}} \Big ) \otimes \ketbra{b}{b}_{\reg{F}} 
      \]
      where $M_1 = \rho^2$ and $M_0 = I - M_1$. 
      \item Furthermore, $\Test$ runs the unitary part\footnote{See \Cref{sec:algorithms} for a definition of the unitary part of a generalized quantum circuit.} of $G$ as a black box, and if the complexity of $G$ is $T$, the complexity of $\Test$ is $O(T + n)$.
  \end{enumerate}
\end{lemma}
\ifcrypto
\else
  \begin{proof}
    \input{prs-testlemma}
  \end{proof}
\fi

We note that if a PRFS satisfies perfect state generation, then the $\Test$ algorithm corresponding to the circuit $G_\secparam(k,x)$ implements a projection onto the state $\ket{\psi_{k,x}} = G_\secparam(k,x)$ in the case that the $\Test$ accepts (i.e. outputs $1$). If the PRFS satisfies the weaker recognizable abort property, we get that the $\Test$ algorithm implements a \emph{scaled} projection onto the correct state $\ket{\psi_{k,x}}$.

\begin{corollary}[PRFS tester with recognizable abort]
\label{lem:test}
  Let $G$ be a $(d,n)$-PRFS generator with the recognizable abort property. Then there exists a QPT algorithm $\Test$ such that for all $\lambda$, $k \in \{0,1\}^\lambda$ and $x \in \{0,1\}^{d(\lambda)}$, for all density matrices $\sigma_{\reg{E} \reg{Q}}$ where $\reg{Q}$ is an $n(\secparam)$-qubit register, applying $\Test(k,x,\cdot)$ to register $\reg{Q}$ yields the following state on registers $\reg{E} \reg{F}$ where $\reg{F}$ stores the decision bit:
   \[
        (I_{\reg{E}} \otimes \Test_{\reg{Q}}) (k,x,\sigma_{\reg{E} \reg{Q}}) = \sum_b \Tr_{\reg{Q}} \Big( (I_\reg{E} \otimes M_b) \sigma_{\reg{E} \reg{Q}} \Big ) \otimes \ketbra{b}{b}_{\reg{F}} 
      \]
      where $M_1 = \eta^2 \ketbra{\psi}{\psi}$ and $M_0 = I - M_1$ with $\eta, \ket\psi$ (which generally depend on $k,x$) are those guaranteed by the recognizable abort property.
\end{corollary}
\begin{proof}
  Fix $\lambda$ and $k\in\{0,1\}^\lambda, x \in \{0,1\}^{d(\lambda)}$.
  By the recognizable abort property, we know that $G_\secparam(k, x) = \eta\ketbra\psi\psi + (1 - \eta)\ketbra\bot\bot$.
  We implement the circuit $\Test$ by first testing whether the input state is $\ket\bot$ (which we can do since it is a fixed known state), rejecting if so, and otherwise applying the test circuit from \Cref{prop:test-channel} with the circuit $G_{k,x}$ that takes no input and outputs $\rho = G_\secparam(k, x)$.
  Since we projected the input state to have no overlap with $\ket{\bot}$, we get that
  \[
    \rho \, \sigma \, \rho = \eta^2 \, \ketbra{\psi}{\psi} \, \sigma \, \ketbra{\psi}{\psi}
  \]
  as desired.
\end{proof}

Next we analyze a \emph{product} of $\Test$ algorithms run in parallel on different qubits of a (possibly entangled) state.

\begin{corollary}[Product of PRFS testers with recognizable abort]
\label{lem:test-product}
  Let $G$ be a $(d, n)$-PRFS generator with the recognizable abort property and let $\Test$ denote the corresponding tester algorithm given by \Cref{lem:test}. Fix $\lambda, t \in \N$. For all $k_1,\ldots,k_t \in \{0,1\}^\lambda$ and for all $x_1,\ldots,x_t \in \{0,1\}^{d(\lambda)}$, define the QPT algorithm $\Test^{\otimes t}$ that given an $t\cdot n(\lambda)$-qubit density matrix $\sigma$ behaves as follows: for all $i=1,\ldots,t$, on the $i$'th block of $n(\lambda)$ qubits of $\sigma$, run the algorithm $\Test_\lambda(k_i,x_i,\cdot)$. Output $1$ if and only if all $t$ invocations of $\Test$ output $1$.
  
  Then $\Test^{\otimes t}$ satisfies the following. For all density matrices $\sigma_{\reg{E} \reg{Q}}$ where $\reg{Q}$ is an $t \cdot n(\secparam)$-qubit register, applying $\Test^{\otimes t}$ to register $\reg{Q}$ yields the following state on registers $\reg{E} \reg{Q} \reg{F}$ where $\reg{F}$ stores the decision bit:
   \[
        (I_{\reg{E}} \otimes \Test^{\otimes t}) (\sigma_{\reg{E} \reg{Q}}) = \sum_b \Tr_{\reg{Q}} \Big( (I_\reg{E} \otimes M_b) \sigma_{\reg{E} \reg{Q}} \Big ) \otimes \ketbra{b}{b}_{\reg{F}} 
      \]
where $M_1 = \eta^2 \ketbra{\psi}{\psi}$ and $M_0 = I - M_1$ with $\ket{\psi} = \ket{\psi_{k_1,x_1}} \otimes \cdots \otimes \ket{\psi_{k_t,x_t}}$, and $\eta = \eta_{k_1,x_1} \cdots \eta_{k_t,x_t}$ where $\ket{\psi_{k_i, x_i}}, \eta_{k_i, x_i}$ for $i = 1, ..., t$ are the values guaranteed by the recognizable abort property.
\end{corollary}
\begin{proof}
  This follows from the fact that each invocation of $\Test(k_i,x_i,\cdot)$, conditioned on accepting, implements a (scaled) projection $\eta_{k_i,x_i} \ketbra{\psi_{k_i,x_i}}{\psi_{k_i,x_i}}$ on a disjoint register of $\sigma$. 
\end{proof}

We note that the previous two Corollaries establish the behavior of the $\Test$ procedure for \emph{every} fixed key $k$ (or sequence of keys, in the case of \Cref{lem:test-product}). The next Lemma establishes the behavior of the $\Test$ procedure when given outputs of \emph{any} PRFS generator (even ones without recognizable abort); the bounds are stated \emph{on average} over a uniformly random key $k$.

\begin{lemma}[Self-testing PRFS]
  \label{lem:test-honest}
  Let $G$ be a $(d, n)$-PRFS generator and $\Test(k, x, \cdot)$ denote the tester algorithm for $G(k, x)$ given by \Cref{prop:test-channel}.
  There exists a negligible function $\nu(\cdot)$ such that for all $\lambda$, for all $x \neq y$,
  \[
    \Pr_{k}[\Test(k, x, G(k, x)) = 1] \ge 1 - \nu(\lambda),
  \]
  and
  \[
    \Pr_{k}[\Test(k, x, G(k, y)) = 1] \le 2^{-n(\lambda)} + \nu(\lambda).
  \]
\end{lemma}
\begin{proof}
  By \Cref{prop:test-channel},
  \[
    \Pr_{k}[\Test(k, x, G(k, x)) = 1]
      = \E_k\left[\Tr(G(k, x)^3)\right]
      \ge \E_k\left[\frac{\Tr(G(k, x)^2)^2}{\Tr(G(k, x))}\right]
      = \E_k\left[\Tr(G(k, x)^2)^2\right],
  \]
  which is negligibly close to 1 by Item 1 of \Cref{lem:prs-orthog} and Markov's inequality, and the inequality is due to the following fact.
  For any finite-dimensional vector $x$ with non-negative coefficients, by Cauchy--Schwarz we have
  \begin{align*}
    \norm{x}_1 \norm{x}_3^3
      &= \left(\sum_i x_i\right) \left(\sum_i x_i^3\right)
      \ge \left(\sum_i \sqrt{x_i} \sqrt{x_i^3}\right)^2
      = \norm{x}_2^4.
  \end{align*}
  Similarly,
  \[
    \Pr_{k}[\Test(k, x, G(k, y)) = 1]
      = \E_k\left[\Tr(G(k, x)^2G(k, y))\right]
      \le \E_k\left[\Tr(G(k, x)G(k, y))\right],
  \]
  which is at most negligibly larger than $2^{-n(\lambda)}$ by Item 2 of \Cref{lem:prs-orthog}, and the inequality is due to the fact that $G(k, x) \preccurlyeq I$.
\end{proof}

%% file: prs-properties.tex
\subsection{Basic Properties of PRS and PRFS Generators}

In this section we present %
the following Lemma, which establishes some orthogonality and purity properties of the output of PRFS generators, on average over the key. 
\begin{lemma}[Properties of PRFS generator outputs]
\label{lem:prs-orthog}
Let $G$ be a $(d,n)$-PRFS generator. Then there exists a negligible function $\eps(\lambda)$ such that for all $\lambda$, for all $x,y \in \{0,1\}^{d(\lambda)}$ where $x \neq y$, we have
\begin{enumerate}
    \item $
    \E_{k \leftarrow \{0,1\}^\lambda} \, \Tr(G_\lambda(k,x) \, G_\lambda(k,y)) \leq 2^{-n(\lambda)} + \eps(\lambda)$;
\item $
    \E_{k \leftarrow \{0,1\}^\lambda} \, \Tr(G_\lambda(k,x)^2) \ge 1 - \eps(\lambda)$.
\end{enumerate}
\end{lemma}
\begin{proof}
Consider the following QPT algorithm $A$: on input $(x,y,\ket{\phi_1},\ket{\phi_2})$, it performs the SWAP test on $\ket{\phi_1}$ and $\ket{\phi_2}$ and accepts if the SWAP test accepts. 
If $\ket{\phi_1},\ket{\phi_2}$ are independently sampled according to the Haar measure on $n(\lambda)$ qubits, the acceptance probability is on average
\begin{equation}
    \label{eq:prs-orthog-1}
    \frac{1}{2} + \frac{1}{2} \E_{\ket{\phi_1},\ket{\phi_2} \leftarrow \Haar_{n}} \left | \ip{\phi_1}{\phi_2} \right |^2 = \frac{1}{2} + \frac{1}{2} 2^{-n(\lambda)}
\end{equation}
where we used \Cref{fact:haar-avg}. On the other hand, if the algorithm $A$ is run on input $(x,y,G_\lambda(k,x),G_\lambda(k,y))$ for randomly chosen $k$ the acceptance probability is on average
\begin{equation}
    \label{eq:prs-orthog-2}
\frac{1}{2} + \frac{1}{2} \E_{k \leftarrow \{0,1\}^\lambda}   \Tr(G_\lambda(k,x) \, G_\lambda(k,y))~.
\end{equation}
Since $A$ is a QPT algorithm, by the pseudorandomness property of the PRFS generator, \Cref{eq:prs-orthog-1,eq:prs-orthog-2} are negligibly different. Specifically, their difference is $\eps(\secparam)$, where $\eps(\secparam)$ is the negligible function guaranteed by the pseudorandomness property. This implies the first item of the Lemma.

For the second item of the Lemma, if $\ket{\phi_1} = \ket{\phi_2}$, then the algorithm $A$ accepts $(x,x,\ket{\phi_1},\ket{\phi_1})$ with probability $1$. 
On the other hand, if the algorithm $A$ is run on input $(x,x,G_\lambda(k,x),G_\lambda(k,x))$, then the acceptance probability is on average
\[
\frac{1}{2} + \frac{1}{2} \E_{k \leftarrow \{0,1\}^\lambda}   \Tr(G_\lambda(k,x)^2)~.
\]
Since the algorithm is efficient and only uses the output of the generator instead of the key, this implies that $\E_{k \leftarrow \{0,1\}^\lambda}   \Tr(G_\lambda(k,x)^2)$ is negligibly (specifically, $\eps(\secparam)$) different from $1$, as desired. 
\end{proof}

%% file: prs-testlemma.tex
Consider the unitary part $\hat{G}$ of $G$, which takes as input a register $\reg{A}$ and outputs registers $\reg{R} \reg{B}$ where $\reg{R}$ has $n$-qubits and $\reg A$ and $\reg B$ have the appropriate number of qubits.
  
The circuit $\Test$ takes as input an $n$-qubit register $\reg{Q}$ and outputs registers $\reg{F} \reg{Q}$ where $\reg{F}$ is a single-qubit accept/reject register. %
It behaves as follows:
\begin{enumerate}
  \item \label{alg:test-channel-0} Initialize an ancilla register $\reg{A}$ in the state $\ket{0 \cdots 0}$, and initialize a single-qubit register $\reg{F}$ in the state $\ket{0}$.
  \item \label{alg:test-channel-1} Run the unitary part $\hat{G}$ on register $\reg A$ to obtain registers $\reg{R} \reg{B}$;
\item \label{alg:test-channel-2} Swap the registers $\reg{Q}$ and $\reg{R}$;
\item \label{alg:test-channel-3} Apply the inverse $\hat{G}^\dagger$ on registers $\reg{R} \reg{B}$ to get register $\reg{A}$;
\item \label{alg:test-channel-4} Measure the register $\reg{A}$ in the computational basis; if the outcome is $\ket{0 \cdots 0}$, then flip the qubit in $\reg{F}$ to $\ket{1}$. 
\item \label{alg:test-channel-5} 
Trace out the register $\reg{Q}$. 
\end{enumerate}
This concludes the description of $\Test$. Item 2 of the Lemma statement follows from inspection. 
\medskip
\begin{figure}[H]
\centering
$
    \Qcircuit @C=0.6em @R=1.0em {
      \sigma & & &  & & {/} \qw & \qw & \multigate{1}{S}  & \qw & \qw & \qw & \qswap & &  & & & \reg{Q} \\
      \ket{0 \cdots 0} & & & & & {/} \qw  & \gate{\hat{G}} & \ghost{S} & \gate{\hat{G}^\dagger} & \qw & \meter & \control \cw \cwx[1]  & & & & & \reg{A} \\
      \ket{0}  & & & &  & \qw & \qw & \qw     & \qw & \qw & \qw & \targ & \cw & \cw & & & \reg{F}
    }
    $
\caption{$\Test$ circuit. The $S$ gate denotes SWAP between registers $\reg{Q}$ and $\reg{R}$. The $\reg{F}$ register is set to $\ket{1}$ if and only if the $\reg{A}$ register measures to all zeroes. The $\reg{Q}$ and $\reg{A}$ registers are traced out at the end.}
\label{fig:teleportation-gadget-intro-simple}
\end{figure}
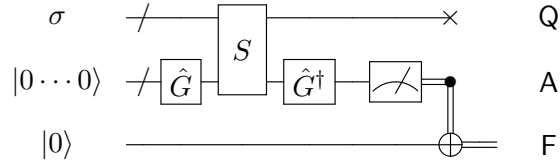
\medskip

We now prove Item 1 of the Lemma statement. Fix a density matrix $\sigma_{\reg{E} \reg{Q}}$. Without loss of generality we can assume that $\sigma$ is a pure state $\ket{\theta}_{\reg{E} \reg{Q}}$ (because we can let $\reg{E}$ contain the purification). We analyze running the circuit $\Test$ on $\reg{Q}$ of $\ket{\theta}$.

Let $\ket{\theta}_{\reg{E} \reg{Q}} = \sum_i \alpha_i \ket{u_i}_{\reg{E}} \otimes \ket{v_i}_{\reg{Q}}$ denote the Schmidt decomposition of $\ket{\theta}$ for some orthonormal bases $\{ \ket{u_i} \}, \{ \ket{v_i} \}$.
After Step~\ref{alg:test-channel-1} of the $\Test$ circuit, the global state is
\[
    \ket{\theta}_{\reg{E} \reg{Q}} \otimes \hat{G} \ket{0 \cdots 0}_{\reg{A}}~.
\]
Let $\hat{G}\ket{0}_{\reg A} = \sum_j \beta_j \ket{\psi_j}_{\reg R} \otimes \ket{\phi_j}_{\reg B}$ denote the Schmidt decomposition of $\hat{G}\ket{0}$ for some orthonormal bases $\{ \ket{\psi_j} \}, \{ \ket{\phi_j} \}$. 
After Step~\ref{alg:test-channel-2} the global state can be written as
\[
    \sum_{ij} \alpha_i \beta_j \ket{u_i}_{\reg{E}} \otimes \ket{\psi_j}_{\reg{Q}} \otimes \ket{v_i}_{\reg{R}} \otimes \ket{\phi_j}_{\reg{B}}~.
\]
After Step~\ref{alg:test-channel-3} the global state can be written as 
\[
    \sum_{ij} \alpha_i \beta_j \ket{u_i}_{\reg{E}} \otimes \ket{\psi_j}_{\reg{Q}} \otimes \hat{G}^\dagger \Big( \ket{v_i}_{\reg{R}} \otimes \ket{\phi_j}_{\reg{B}} \Big)~.
\]
In Step~\ref{alg:test-channel-4}, the $\reg{A}$ register is measured. If the outcome is all zeroes, then the post-measurement state can be written as (up to normalization) 
\begin{align*}
    &\equad \sum_{ij} \alpha_i \beta_j \ket{u_i}_{\reg{E}} \otimes \ket{\psi_j}_{\reg{Q}} \otimes \ketbra{0}{0}_{\reg{A}} \hat{G}^\dagger \Big( \ket{v_i}_{\reg{R}} \otimes \ket{\phi_j}_{\reg{B}} \Big) \\
    &= \sum_{ij} \alpha_i \beta_j \ket{u_i}_{\reg{E}} \otimes \ket{\psi_j}_{\reg{Q}} \otimes \Big( \sum_k \beta_k \bra{\psi_k}_{\reg{R}} \otimes \bra{\phi_k}_{\reg{B}} \Big) \Big( \ket{v_i}_{\reg{R}} \otimes \ket{\phi_j}_{\reg{B}} \Big) \, \ket{0}_{\reg{A}}  \\
    &= \sum_{ij} \alpha_i \, \beta_j^2 \, \ip{\psi_j}{v_i} \, \ket{u_i}_{\reg{E}} \otimes \ket{\psi_j}_{\reg{Q}} \otimes \ket{0}_{\reg{A}} \\
    &= (I_{\reg{E}} \otimes \rho_{\reg{Q}}) \ket{\theta}_{\reg{E} \reg{Q}} \otimes \ket{0}_{\reg{A}}
\end{align*}
where in the second line we used our Schmidt decomposition for $\hat{G} \ket{0 \cdots 0}$ and in the third line we used the orthonormality of the basis $\{ \ket{\phi_j} \}$. The fourth line follows since by definition of $G$ we have
  $\rho = \sum_j \beta_j^2 \ketbra{\psi_j}{\psi_j}$.

If the measurement outcome is all zeroes, the register $\reg{F}$ is set to $\ket{1}$. Otherwise it remains $\ket{0}$.

In Step~\ref{alg:test-channel-5}, the registers $\reg{Q}$ and $\reg{A}$ are traced out. Thus conditioned on getting the all zeroes outcome, the state on register $\reg{E}$ is 
\[
    \Tr_{\reg{Q}} \Big( \rho_{\reg{Q}} \, \ketbra{\theta}{\theta}_{\reg{E} \reg{Q}} \, \rho_{\reg{Q}} \Big) = \Tr_{\reg{Q}} \Big( \rho_{\reg{Q}}^2 \, \ketbra{\theta}{\theta}_{\reg{E} \reg{Q}}\Big)
\]
where we used the cyclicity of the partial trace with respect to operators acting on register $\reg{Q}$ only. Conditioned on \emph{not} getting the all zeroes outcome, it must be that the state on register $\reg{E}$ is
\[
\Tr_{\reg{Q}} \Big( (I - \rho^2)_{\reg Q} \, \ketbra{\theta}{\theta}_{\reg{E} \reg{Q}}\Big)~.
\]
This establishes that the output of the $\Test$ algorithm is as described in the Lemma statement.

%% file: prfsfromprs.tex
\newcommand{\measurefail}{\mathsf{Abort}}
\newcommand{\swapfail}{\mathsf{SwapFail}}
\newcommand{\orthfail}{\mathsf{OrthFail}}

\section{Constructing PRFS from PRS}
\label{sec:prfs-from-prs}

In this section we present our construction of PRFS generators using PRS generators, which are seemingly weaker objects. As mentioned in the introduction, there is a trivial construction of PRFS from PRS. Let $G$ be a PRS generator. Define the PRFS generator $G'$ with input length $d(\lambda) = O(\log \lambda)$, where $G'_{\lambda'}(k,x) = G_\lambda(k_x)$ with $\lambda' = 2^{d(\lambda)} \lambda$ and $k_x$ denoting the $x$'th block of $\lambda$ bits in $k \in \{0,1\}^{\lambda '}$. 
However, this simple construction is such that the input length is always at most logarithmic in the seed length. This, as far as we can tell, is not very useful for applications. 

We are going to present a more interesting construction: we will build a PRFS generator for \emph{any} input length $d(\secparam)$ that is at most constant times $\log \secparam$, as long as the the output length of the starting PRS generator is at least $2d(\secparam) + \omega(\log \log \secparam)$.
Although the input length may appear modest, such PRFS generators are sufficient for most of the applications we consider in this paper. We find it an intriguing question of whether it is possible to construct PRFS generators with longer input lengths from PRS generators in a black box way.%

\begin{theorem}
  \label{thm:prfs-from-prs}
  Let $d(\secparam),n(\secparam)$ be functions such that $d(\secparam) = O(\log \secparam)$ and $n(\secparam) = d(\secparam) + \omega(\log\log \secparam)$.
  Let $G$ denote a $(n(\secparam) + d(\secparam))$-PRS generator.
  Then there exists a $(d(\secparam),n(\secparam))$-PRFS generator $F$ with the recognizable abort property, such that for all $\secparam$ the circuit $F_\secparam$ invokes the $G_\secparam$ as a black box. 
\end{theorem}

The rest of this section is dedicated to proving the theorem.
For notational clarity we use the abbreviations $d = d(\secparam)$ and $n = n(\secparam)$. %

The construction of the PRFS generator is given by the following circuit $F_{\secparam}(k,x)$.
On input key $k \in \{0,1\}^{\secparam}$, input $x \in \{0,1\}^{d}$, repeat the following $2^d \cdot \secparam$ times: 
        \begin{itemize}
            \item Compute the $(d + n)$-qubit state $\rho_k \leftarrow G_{\secparam}(k)$.

            \item Measure the first $d$ qubits of $\rho_k$ in the computational basis to obtain a string $y \in \{0,1\}^d$. If $y = x$, then output the remaining $n$ qubits. Otherwise, continue.
        \end{itemize}
If the measurement outcomes was different from $x$ in all the $2^d \secparam$ iterations, set $\sigma_{k,x}=\ketbra{\bot}{\bot}$.
Let the output be $\sigma_{k,x}$.

\medskip

The algorithm $F = \{F_\secparam \}_\secparam$ is uniform QPT because for each $\secparam$, the running time of the circuit $F_\secparam$ is going to be $O(2^d \cdot \secparam)$ times the complexity of running $G_\secparam$, which is QPT since $d = O(\log \secparam)$ and $G$ is QPT.
It is easy to see that even if $G$ (as a PRFS generator) only satisfies recognizable abort (instead of perfect generation), $F$ still satisfies recognizable abort by construction.
Therefore, the construction also works with the PRS generator constructed by Brakerski and Shmueli~\cite{BrakerskiS20}.

\ifcrypto
\else
  We now argue that the outputs of $F$ satisfy the pseudorandomness property of a PRFS.
  \input{prfsfromprs-security}

\fi

%% file: prfsfromprs-security.tex
Assume for contradiction that there exists a non-uniform QPT adversary $A$ and distinct $x_1,\ldots, x_{s} \in \{0,1\}^{d}$ such that
 \begin{align}
 \label{eqn:prfs:proof}
        &\Big | \Pr_{k \leftarrow \{0,1\}^\secparam} \left [ A_\secparam( x_1,\ldots,x_{s},F_{\secparam}(k,x_1)^{\otimes t},\ldots, F_{\secparam}(k,x_{s})^{\otimes t}) = 1 \right] \\ 
        & \qquad \qquad - \Pr_{\ket{\vartheta_1}, \ldots,\ket{\vartheta_{s}} \leftarrow \Haar_{n}} \left [ A_\secparam( x_1,\ldots,x_{s}, \ket{\vartheta_1}^{\otimes t},\ldots, \ket{\vartheta_{s}}^{\otimes t}) = 1 \right] \Big | = \eps(\secparam) \nonumber
    \end{align}
is not negligible in $\secparam$. %
Let $M = \secparam s t 2^d = \lambda^{O(1)}$.
We are going to construct an adversary $B$ that breaks the pseudorandomness property for the underlying PRS used in the construction using $M$ copies.
Formally, $B_{\secparam}$ is a QPT algorithm that takes as input $\rho^{\otimes M}$ where $\rho$ is a $(d+n)$-qubit state and does the following: 
\begin{itemize}
\item For $j = 1,\ldots,s$, repeat the following $t$ times:
        \begin{itemize}
            \item Repeat the following $ \secparam 2^{d}$ times: Measure the first $d$ qubits of a new copy of $\rho$ in the computational basis to obtain a string $y \in \{0,1\}^d$. If $y = x_j$, then save the remaining $n$ qubits of $\rho$ (which we denote as the state $\sigma_{x_j}$). %
            Otherwise, continue. 
            \item If the outcome $x_j$ was never measured, $B_{\secparam}$ aborts.
        \end{itemize}
\item Execute $b \leftarrow A_{\secparam}\left( x_1,\ldots,x_{s}, \sigma_{x_1}^{\otimes t},\ldots,\sigma_{x_s}^{\otimes t} \right)$.
\item Output $b$. 
\end{itemize} 
We show the following: 
 \begin{equation}
 \label{eqn:prfs:secred}
        \Big | \Pr_{k \leftarrow \{0,1\}^\secparam} \left [ B_\secparam(G_{\secparam}(k)^{ \otimes M }) = 1 \right]  - \Pr_{\ket{\vartheta} \leftarrow \Haar_{d + n}} \left [ B_{\secparam}(\ket{\vartheta}^{ \otimes M}) = 1 \right] \Big | \geq \eps(\secparam) - \nu(\secparam)
    \end{equation}
for some negligible function $\nu(\secparam)$.
This in turn shows that the algorithm $B = \{ B_\lambda \}_\lambda$ violates the pseudorandomness assumption on the PRS generator $G$, which is a contradiction. Thus $\eps(\lambda)$ must be negligible.

We prove \eqref{eqn:prfs:secred} by the following hybrid argument.
Note that the probability on the left hand side is by construction the same as the probability that $A$ outputs 1 in the real world experiment (when $A$ is given PRFS as input).
Therefore, we arrive at \eqref{eqn:prfs:secred} by triangle inequality after comparing the probability on the right hand side with the probability that $A$ outputs 1 in the ideal world experiment (when $A$ is given Haar random states as input), which is given by the following Lemma.

\begin{lemma}
\label{clm:prfs:int2}
If $d(\secparam) = O(\log \secparam)$ and $n(\secparam) = d(\secparam) + \omega(\log\log \secparam)$, then for any polynomial $s(\cdot), t(\cdot)$, there exists a negligible function $\nu(\secparam)$ such that
\begin{align*}
 \left | \Pr_{\ket{\vartheta} \leftarrow \Haar_{d + n}} \left [ B_{\secparam}(\ket{\vartheta}^{ \otimes M}) = 1 \right] - \Pr_{\ket{\vartheta_1}, \ldots,\ket{\vartheta_{s}} \leftarrow \Haar_{n}} \left [ A_\secparam( x_1,\ldots,x_{s}, \ket{\vartheta_1}^{\otimes t},\ldots, \ket{\vartheta_{s}}^{\otimes t}) = 1 \right] \right | \leq  \nu(\secparam)~.
\end{align*}
\end{lemma}
\begin{proof}
Consider the behavior of the algorithm $B_\lambda$ on input $\ket{\vartheta}^{\otimes M}$ for $\ket{\vartheta}$ sampled from the Haar distribution $\Haar_{d+n}$. Define the distribution $\mathscr{R}$ over $(d+n)$-qubit unitary operators
\[
    R = \sum_{x \in \{0,1\}^d} \ketbra{x}{x} \otimes R_x
\]
where $(R_x)_{x \in \{0,1\}^d}$ is a sequence of i.i.d. Haar-random $n$-qubit unitaries. 

Observe that, by the unitary invariance of the Haar measure, $R \ket{\vartheta}$ is also distributed according to $\Haar_{d+n}$. Therefore the algorithm $B_\lambda$ behaves identically on input $(R \ket{\vartheta})^{\otimes M}$ for any $R$ in the support.

Let $\measurefail$ denote the event that the algorithm $B_\lambda$ aborts on input $(R \ket{\vartheta})^{\otimes M}$; this happens only if there exists a $j\in [s]$ such that, even after measuring the first $d$ qubits of $\lambda t 2^d$ copies of $R \ket{\vartheta}$, the string $x_j$ occured fewer than $t$ times as a measurement outcome. 

Notice that the event $\measurefail$ (and its negation) is \emph{independent} of the choice of randomizing unitaries $(R_x)_x$; that is because applying $R$ to $\ket{\vartheta}$ does not change the distribution of measurement outcomes on the first $d$ qubits. Thus, for all $\ket{\vartheta} = \sum_x \alpha_x \ket{x} \otimes \ket{\vartheta_x}$, conditioning on the event $\neg \measurefail$ (the negation of $\measurefail$) still leaves the unitary $R$ distributed according to $\mathscr{R}$. 

Therefore for all $\ket{\vartheta}$ and for all $R$ we have
\begin{equation}
    \label{eq:prs:int2:1}
\Pr \left [ B_\lambda((R \ket{\vartheta})^{\otimes M}) = 1 \big | \, \neg \measurefail \right ] = \Pr \left [ A_\secparam( x_1,\ldots,x_{s}, (R_{x_1} \ket{\vartheta_{x_1}})^{\otimes t}, \ldots,(R_{x_s} \ket{\vartheta_{x_s}})^{\otimes t}) = 1 \right ]
\end{equation}
where the probabilities are over the randomness of the measurements.
Therefore, \eqref{eq:prs:int2:1} also holds if the probability also averages over the randomness of sampling $R \leftarrow \mathscr{R}$.
Since the $R_{x_j}$'s are i.i.d. Haar-random unitaries and the $x_i$'s are distinct, we conclude that \eqref{eq:prs:int2:1} is exactly equal to the probability $A$ outputs 1 in the ideal experiment.
Thus 
\begin{align*}
&\equad \Pr \left [ B_\lambda( \ket{\vartheta}^{\otimes M}) = 1 \right ] \\
&= \Pr \left [ B_\lambda((R \ket{\vartheta})^{\otimes M}) = 1 \right ] \\
&= \Pr \left [  \measurefail \right ] \cdot \Pr \Big [ B_\lambda((R \ket{\vartheta})^{\otimes M}) = 1 \big | \, \measurefail \Big ] + \Pr \left [ \neg \measurefail \right ] \cdot \Pr \Big [ B_\lambda((R \ket{\vartheta})^{\otimes M}) = 1 \big | \, \neg \measurefail \Big ]   \\
&= \Pr \left [  \measurefail \right ] \cdot \Big(\Pr \Big [ B_\lambda((R \ket{\vartheta})^{\otimes M}) = 1 \big | \, \measurefail \Big ]  - \Pr \Big [ B_\lambda((R \ket{\vartheta})^{\otimes M}) = 1 \big | \, \neg \measurefail \Big ] \Big) \\
&\qquad \qquad + \Pr \Big [ B_\lambda((R \ket{\vartheta})^{\otimes M}) = 1 \big | \, \neg \measurefail \Big ] 
\end{align*}
where in the last equality we used $\Pr \left [  \neg \measurefail \right ] = 1 - \Pr \left [  \measurefail \right ]$. 
Thus combining this with \eqref{eq:prs:int2:1},
\begin{align*}
  &\equad \left | \Pr \left [ B_{\secparam}(\ket{\vartheta}^{ \otimes M}) = 1 \right] - \Pr \left [ A_\secparam( x_1,\ldots,x_{s}, (R_{x_1} \ket{\vartheta_{x_1}})^{\otimes t}, \ldots,(R_{x_s} \ket{\vartheta_{x_s}})^{\otimes t}) = 1 \right ] \right | \\
 &\leq \Pr \left [  \measurefail \right ] \cdot \Big | \Pr \Big [ B_\lambda((R \ket{\vartheta})^{\otimes M}) = 1 \big | \, \measurefail \Big ] - \Pr \Big [ B_\lambda((R \ket{\vartheta})^{\otimes M}) = 1 \big | \, \neg \measurefail \Big ] \Big | \\
 &\leq \Pr \left [  \measurefail \right ].
\end{align*}

We now estimate the probability of the event $\measurefail$. Fix a state $\ket{\vartheta}$ and let $p_{x_j}$ denote the probability of obtaining $x_j$ when measuring the first $d$ qubits of $\ket{\vartheta}$, or equivalently since $R$ commutes with the measurement, $R \ket{\vartheta}$. Fix a $j \in [s]$. The probability that measuring $\lambda 2^d$ copies of $R \ket{\vartheta}$ fails to yield the outcome $x_j$ is equal to
\[
    \Big ( 1 - p_{x_j} \Big) ^{\secparam 2^d}
\]

The algorithm aborts if this happens in any of the $st$ iterations of the ``main loop'' of $B_\lambda$; thus the probability of $\measurefail$ is, by union bound, at most
\[
    \sum_{j=1}^s t \Big ( 1 - p_{x_j} \Big) ^{\secparam 2^d}
\]
The following Lemma establishes deviation bounds on the probabilities $p_x$:

\begin{lemma}
\label{lem:haar-prefix}
Let $\ket{\psi}$ be sampled from the Haar distribution $\Haar_{d+n}$. For all $x \in \{0,1\}^d$, let $p_x$ denote the probability of measuring the first $d$ qubits of $\ket{\psi}$ in the computational basis and obtaining outcome $x$. Then for all $\delta > 0$ with probability at least $1 - 2^{d+1} \cdot \exp(-C 2^{n+d} \delta^2 )$ over $\ket{\psi}$ for some universal constant $C > 0$, we have that
\[
    |p_x - 2^{-d}| \leq \delta 
\]
for all $x \in \{0,1\}^d$.
\end{lemma}

By \Cref{lem:haar-prefix} (setting $\delta = 2^{-d}/2$), with all but negligible probability over the choice of $\ket{\vartheta}$, each of the $p_{x_j}$'s are at least $2^{-d}/2$.
Therefore by union bound, the probability of $\measurefail$, when averaged over the choice of $\ket{\vartheta}$, is at most
\[
    st(1 - 2^{-d}/2)^{\lambda 2^d} + 2\exp(-(C2^{n - d} - d)) \leq st \exp(-\Omega(\lambda)) + 2\exp(-(C2^{n - d} - d))
\]
which for our choice of $s,t,n,d$ is negligible in $\lambda$. 
\end{proof}

\begin{proof}[Proof of \Cref{lem:haar-prefix}]
We first show that, with high probability over $\ket{\psi}$, the probability obtaining any \emph{fixed} prefix $x \in \{0,1\}^d$ is going to be exponentially small in $2^n$. We then apply a union bound over all $x \in \{0,1\}^d$ to obtain the Lemma statement.

Let $\Pi_x$ denote the projector onto the first $d$ qubits being in the state $\ket{x}$. Define $p_x = \Tr(\Pi_x \ketbra{\psi}{\psi})$. On average over the choice of $\ket{\psi}$, this quantity is equal to
\[
    \E_{\ket{\psi} \leftarrow \Haar_{d+n}} p_x = \Tr \left (\Pi_x \, \E_{\ket{\psi} \leftarrow \Haar_{d+n}} \ketbra{\psi}{\psi} \right ) = 2^{-(d+n)}\,  \Tr( \Pi_x) = 2^{-d}
\]
where we used the fact that the average of a Haar-random state is the maximally mixed state.

We now appeal to \emph{L\'{e}vy's Lemma} (\Cref{fact:levy}), which shows that $p_x$ concentrates tightly around its expectation. Define $f(\ket{\psi}) = \Tr( \Pi_x \, \ketbra{\psi}{\psi})$. We calculate the Lipschitz constant of $f$: %
\begin{align*}
    &\equad \frac{|f(\ket{\psi}) - f(\ket{\phi})|}{\norm{\ket{\psi} - \ket{\phi}}_2} \\
    &= \frac{\left |\Tr \Big (\Pi_x ( \ketbra{\psi}{\psi} - \ketbra{\phi}{\phi}) \Big ) \right |}{\norm{\ket{\psi} - \ket{\phi}}_2} \\
    &\leq \frac{ \norm{ \ketbra{\psi}{\psi} - \ketbra{\phi}{\phi}}_1}{\norm{\ket{\psi} - \ket{\phi}}_2} \\
    &= \frac{2 \sqrt{1 - | \langle \psi \mid \phi \rangle|^2}}{\norm{\ket{\psi} - \ket{\phi}}_2} \\
    &= \frac{2 \sqrt{\Big(1 - | \langle \psi \mid \phi \rangle| \Big)\Big(1 + | \langle \psi \mid \phi \rangle| \Big) }}{\norm{\ket{\psi} - \ket{\phi}}_2} \\
    &\leq \frac{2 \sqrt{2 \Big(1 - \Re \langle \psi \mid \phi \rangle \Big) }}{\norm{\ket{\psi} - \ket{\phi}}_2} \\
    &\leq 2
\end{align*}
for all $\ket{\psi},\ket{\phi}$.
By \Cref{fact:levy}, we have
\[
    \Pr \left [ \Big | p_x - 2^{-d} \Big| \geq \delta \right ] \leq 2\exp \Big( -C2^{d+n} \delta^2 \Big)
\]
for some universal constant $C > 0$, where the probability is over $\ket{\psi} \leftarrow \Haar_{d+n}$. 
\end{proof}

%% file: encryption.tex
\section{Quantum Pseudo One-Time Pad from PRFS}
\label{sec:qotp}

The first application of PRFS we present is the Quantum Pseudo One-Time Pad (QP-OTP). In classical cryptography, a pseudo one-time pad is like the one-time pad except the key length is shorter than the length of the plaintext message. %
This is often presented in introductory cryptography courses as a basic example of using pseudorandomness to achieve a cryptographic task that is impossible in the information-theoretic setting. Here, we use a PRFS in place of a PRG to encrypt (classical) messages.

We point out that without knowing about the notion of PRFS, it appears difficult and challenging to construct secure quantum one-time pad schemes directly from PRS generators alone. 

\begin{definition}[Quantum Pseudo One-Time Pad]
\label{def:qpotp}
We say that a pair of QPT algorithms $(\Enc,\Dec)$ is a \emph{quantum pseudo one-time pad (QP-OTP)} for messages of length $\ell(\lambda) > \lambda$ for some polynomial $\ell(\cdot)$ if the following properties are satisfied:
\begin{itemize}
    \item \textbf{Correctness}: There exists a negligible function $\eps(\cdot)$ such that for every $\lambda$, every $x \in \{0,1\}^{\ell}$,
    $$\Pr_{\substack{k \leftarrow \{0,1\}^\lambda,\\ \sigma \leftarrow \Enc_\lambda(k,x)}} \left[ \Dec_\lambda(k,\sigma) = x \right] \geq 1 - \eps(\lambda).$$
    
    \item \textbf{Security}: There exist a polynomial $n(\cdot)$ such that for every polynomial $t(\cdot)$, for every nonuniform QPT adversary $A$, there exists a negligible function $\eps(\cdot)$ where for every $\lambda$ and $x \in \{0,1\}^{\ell}$,
    $$\left| \Pr_{\substack{k \leftarrow \{0,1\}^\lambda,\\ \sigma \leftarrow \Enc_\lambda(k,x)}} \left[ A_\lambda(\sigma^{\otimes t}) = 1 \right] - \Pr_{\ket{\vartheta_1},\ldots,\ket{\vartheta_\ell} \leftarrow \Haar_n} \left[ A_\lambda((\ket{\vartheta_1} \otimes \cdots \otimes \ket{\vartheta_\ell})^{\otimes t}) = 1 \right] \right| \leq \eps(\lambda),$$
    where we have abbreviated $n = n(\lambda)$, $\ell = \ell(\lambda)$, and $t = t(\lambda)$.
\end{itemize}
\end{definition}

Here the security holds even if the adversary could see multiple copies of the same ciphertexts, which might be useful for certain applications, for example when the communication channel is adversarially lossy.
However, when $t = 1$, we can see that the security implies that the ciphertext is computationally indistinguishable to random bit strings of length $\ell n$ (or a maximally mixed state) by \Cref{fact:haar-avg}.

To construct such a quantum pseudo one-time pad, let $G$ be a $(d(\lambda),n(\lambda))$-PRFS generator where $d(\lambda) \geq \lceil \log \ell(\lambda) \rceil + 1$ and $n(\lambda) = \omega(\log \lambda)$. We interpret $G_\lambda(k,\cdot)$ as taking inputs of the form $(i,b)$ where $i \in [\ell(\lambda)]$ and $b \in \{0,1\}$. Let $\Test$ denote the test algorithm from \Cref{lem:test-honest}. 

Fix $\lambda$ and let $\ell = \ell(\lambda)$, $d = d(\lambda)$, and $n = n(\lambda)$.

\begin{enumerate}
    \item $\Enc_\lambda(k,x)$: on input $k \in \{0,1\}^\lambda$ and a message $x \in \{0,1\}^{\ell}$, do the following: 
    \begin{itemize}
        \item For every $i \in [\ell]$, compute $\sigma_i \leftarrow G_\lambda(k,(i,x_i))$.
        \item Set $\sigma = \sigma_1 \otimes \cdots \otimes \sigma_\ell$. 
        \item Output the ciphertext state $\sigma$.
    \end{itemize}

    \item $\Dec_\lambda(k,\sigma)$: on input $k$, $\ell n$-qubit ciphertext state $\sigma$, perform the following operations: 
    \begin{itemize}
        \item Parse $\sigma$ as $\sigma_1 \otimes \cdots \otimes \sigma_\ell$.
        \item For $i \in [\ell]$, execute $\Test(k,(i,0),\sigma_i)$. If it accepts, set $x_i = 0$. Otherwise, set $x_i = 1$.
        \item Output $x = x_1 \cdots x_{\ell}$. 
    \end{itemize}
\end{enumerate}

\begin{lemma}
\label{lem:qopt-correctness}
$(\Enc,\Dec)$ satisfies the correctness property of a quantum pseudo one-time pad according to \Cref{def:qpotp}.
\end{lemma}
\begin{proof}
Fix $\lambda$ and let $\ell = \ell(\lambda)$. Fix a message $x \in \{0,1\}^{\ell}$.
Let $\sigma_{k,i} = G_\lambda(k,(i,x_i))$ and let $\sigma_k = \sigma_{k,1} \otimes \cdots \otimes \sigma_{k,\ell}$. %

Consider the decryption process. Fix an index $i \in [\ell]$.
By \Cref{lem:test-honest}, the probability that $\Test \Big (k,(i,0),\sigma_{k,i} \Big)$ accepts (on average over $k$) is negligibly close to $1$ if $x_i = 0$, and it is negligibly close to $0$ if $x_i = 1$, on average over the key $k$ (here we use the fact that the output length of the PRFS generator is $\omega(\log \lambda)$, so that $2^{-n(\lambda)}$ is negligible). Thus the probability that the correct bit $x_i$ gets decoded is negligibly close to $1$. Taking a union bound over all indices $i \in [\ell]$, we get that the probability of decoding $x$ is negligibly close to $1$, over the randomness of the key $k$ and the decryption algorithm.
\end{proof}

\begin{lemma}
  $(\Enc,\Dec)$ satisfies the security property of quantum pseudo one-time pad according to \Cref{def:qpotp}.
\end{lemma}
\begin{proof}
We prove the security via a hybrid argument. Let $n(\lambda)$ denote the output length of the PRFS generator $G$. Fix $\lambda$, and let $\ell = \ell(\lambda)$, $n = n(\lambda)$, and $t = t(\lambda)$. Fix a message $x \in \{0,1\}^{\ell}$. Consider a nonuniform QPT adversary $A$ such that $A_\lambda$ takes as input $t$ copies of an $\ell n$-qubit density matrix $\sigma$. 

\paragraph{Hybrid $\hybrid_1$.} Sample $k \leftarrow \{0,1\}^\lambda$. Compute $\sigma \leftarrow \Enc_\lambda(k,x)$. The output of the hybrid is the output of the adversary $A_\lambda$ on input $\sigma^{\otimes t}$.

\paragraph{Hybrid $\hybrid_2$.} Consider the following QPT algorithm $B_\lambda$: it takes as input $(i_1,b_1),\ldots,(i_\ell,b_\ell) \in [\ell] \times \{0,1\}$ and a $tn$-qubit state $\sigma_1^{\otimes t} \otimes \cdots \otimes \sigma_\ell^{\otimes t}$. The algorithm $B$ runs the adversary $A_\lambda$ on input $(\sigma_1 \otimes \cdots \otimes \sigma_\ell)^{\otimes t}$ and returns A$_\secparam$'s output.

Sample $k \leftarrow \{0,1\}^\lambda$. Compute $t$ copies of $\sigma \leftarrow \Enc_\lambda(k,x)$. 
The output of this hybrid is the output of $B_\lambda$ on input $((1,x_1),\ldots,(\ell,x_\ell))$ and $\sigma^{\otimes t} = \sigma_1^{\otimes t} \otimes \cdots \otimes \sigma_\ell^{\otimes t}$.

\paragraph{Hybrid $\hybrid_3$.}

Sample $t$ copies of Haar-random states $\ket{\vartheta_1},\ldots,\ket{\vartheta_\ell} \leftarrow \Haar_n$. The output of this hybrid is the output of $B_\lambda$ on input $((1,x_1),\ldots,(\ell,x_\ell))$ and $\ket{\vartheta_1}^{\otimes t} \otimes \cdots \otimes \ket{\vartheta_\ell}^{\otimes t}$.

\medskip 

We now argue the indistinguishability of the hybrids. Clearly, hybrids $\hybrid_1$ and $\hybrid_2$ are identical by construction (the adversary $B_\lambda$ ignores its first input and runs $A_\lambda$ on input $\sigma^{\otimes t}$). Hybrids $\hybrid_2$ and $\hybrid_3$ are indistinguishable because of the pseudorandomness property of the PRFS generator $G$. Notice that, by construction, the output of hybrid $\hybrid_3$ is $A_\lambda((\ket{\vartheta_1} \otimes \cdots \otimes \ket{\vartheta_\ell})^{\otimes t})$.
\end{proof}

%% file: commitment.tex
\section{Quantum Bit Commitments from PRFS}
\label{sec:commitment}

\subsection{Definition}
\label{sec:comm:def}

We consider the notion of quantum commitment scheme with statistical binding and computational hiding property. This is analogous to a classical commitment scheme where the messages are allowed to be quantum states.
We in particular focus on bit commitments where the the committed message is a single bit.
We can generically achieve commitments of long messages by composing many instantiations of the bit-commitment scheme in parallel.

A (bit) commitment scheme is given by a pair of (uniform) QPT algorithms $(C,R)$, where $C=\{C_{\secparam}\}_{\secparam \in \mathbb{N}}$ is called the \emph{committer} and $R=\{R
_{\secparam}\}_{\secparam \in \mathbb{N}}$ is called the \emph{receiver}. There are two phases in a commitment scheme: a commit phase and a reveal phase. %
\begin{itemize}
    \item In the (possibly interactive) commit phase between $C_{\secparam}$ and $R_{\secparam}$, the committer $C_{\secparam}$ commits to a bit, say $b$. We denote the execution of the commit phase to be $\sigma_{CR} \leftarrow \commit \langle C_{\secparam}(b),R_{\secparam} \rangle$, where $\sigma_{CR}$ is a joint state of $C_{\secparam}$ and $R_{\secparam}$ after the commit phase. 
    \item In the %
    reveal phase $C_{\secparam}$ interacts with $R_{\secparam}$ and the output is a trit $\mu \in \{0,1,\bot \}$ indicating the receiver's output bit or a rejection flag. We denote an execution of the reveal phase where the committer and receiver start with the joint state $\sigma_{CR}$ by $\mu \leftarrow \reveal \langle C_\lambda, R_\lambda, \sigma_{CR} \rangle$. 
    
\end{itemize}

\noindent We define the properties satisfied by a commmitment scheme. %

\paragraph{Statistical Binding.} We start by discussing the statistical binding property.
The classical statistical binding property could be rephrased as the following in the quantum setting: for any adversarial (possibly unbounded) committer $C^*_{\secparam}$, we require that at the end of the commit phase, with high probability over the measurement randomness of the receiver, there is a unique bit that $C^*_{\secparam}$ can decommit to in the reveal phase.
Unfortunately, this idealistic notion is not always possible to achieve: in some quantum commitment protocols where the receiver does not measure everything, it is possible for $C^*_{\secparam}$ to send a uniform superposition of commitments of 0 and 1 and later can open to either 0 or 1 with equal probability. This attack was observed and taken into account in many works, including but not limited to~\cite{YWLQ15,Unruh16,FUYZ20,BB21}.

To account for this issue, we consider a notion where an extraction procedure can be applied on the state of the receiver after the commit phase. The output is the receiver's post-extraction state along with the extracted bit $b$.
We revise the statistical binding property guarantee to informally require the following: (a) whether the extractor is applied or not is imperceivable to the committer and (b) the committer can almost never decommit to $1 - b$ if the extracted bit is $b$.

\begin{definition}[Statistical Binding]
\label{def:stat:binding}
We say that a quantum commitment scheme $(C,R)$ satisfies \emph{statistical binding} if for any (non-uniform) adversarial committer $C^*=\{C^*_{\secparam}\}_{\secparam \in \mathbb{N}}$, there exists a (possibly inefficient) extractor algorithm ${\cal E}$ such that the following holds: 
$$ \TraceDist{\realexpt_{\secparam}^{C^*}}{ \idealexpt_{\secparam}^{C^*,{\cal E}}} \leq \nu(\secparam),$$
for some negligible function $\nu(\secparam)$, where the experiments $\realexpt_{\secparam}^{C^*}$ and $\idealexpt_{\secparam}^{C^*,{\cal E}}$ are defined as follows.
\begin{itemize}
    \item $\realexpt_{\secparam}^{C^*}$: Execute the commit phase to obtain the joint state $\sigma_{C^* R} \leftarrow \commit\langle C_{\secparam}^*,R_{\secparam} \rangle$. Execute the reveal phase to obtain the trit $\mu \leftarrow \reveal \langle C_{\secparam}^*, R_\secparam, \sigma_{C^* R} \rangle$. Output the pair $(\tau_{C^*},\mu)$ where $\tau_{C^*}$ is the final state of the committer. %
    \item $\idealexpt_{\secparam}^{C^*,{\cal E}}$: Execute the commit phase to obtain the joint state $\sigma_{C^* R} \leftarrow \commit\langle C_{\secparam}^*,R_{\secparam} \rangle$.  Apply the extractor $I \otimes {\cal E}$ on $\sigma_{C^* R}$ (acting \emph{only on the receiver's part}) to obtain a new joint committer-receiver state $\sigma'_{C^* R}$ along with $b' \in \{0,1,\bot\}$.
    Execute the reveal phase 
    to obtain the trit $\mu \leftarrow \allowbreak\reveal \langle C_{\secparam}^*,\allowbreak R_\secparam,\allowbreak \sigma'_{C^* R} \rangle$. Let $\tau_{C^*}$ denote the final state of the committer.
    If $\mu = \bot$ or $\mu = b'$, then output $(\tau_{C^*},\mu)$.
    Otherwise, output a special symbol $\mathfrak{E}$ (unused in the real experiment) indicating extraction error.
    
\end{itemize}
\end{definition}

\begin{remark}
Many prior works consider statistical binding for quantum commitments. We highlight the main differences between our definition and the prior notions. %
\begin{itemize}
    \item Comparison with~\cite{YWLQ15,Unruh16,FUYZ20}: the statistical binding property is formalized by requiring the states of the (honest) committer when committing to bits 0 and 1 to be far in trace distance. While their definition is cleaner (and probably equivalent to our notion), in our opinion, it is unwieldy to use their definition for applications. Specifically, one has to either implicitly or explicitly come up with an extractor in the security proofs for applications \cite{YWLQ15,FUYZ20} and moreover, show that the indistinguishability of the real and the ideal world holds against dishonest committers. On the other hand, we incorporate these technical difficulties as requirements in our definition making it easier to use in applications.
    
    Another downside of the statistical binding property there is that in order for the sum-binding property to be useful in applications, it is common to additionally require the opening phase to follow the ``canonical'' opening protocol, where the committer sends the purification of the mixed state sent in the committing phase, and the receiver performs a rank-1 projection to check the state.
    This implies that \emph{both} parties must keep their part of the state coherent between the two phases. However, our definition gives the flexibility of the reveal phase having purely classical communication.

    \item Comparison with~\cite{BB21}: A related work by~\cite{BB21} considers statistical binding of quantum commitments called classical binding. %
    The main difference is the following. In their notion, the honest receiver applies a measurement that collapses the commitment into a quantum state and a classical string in such a way that the classical string information theoretically determines the message.
    They then use this feature to show that in some applications, the opening of the commitment can be classical.
    Our definition is also more general in the sense that the honest receiver is not required to do any measurement and the collapsing happens implicitly in the ideal world during the execution of extractor.
\end{itemize}
\end{remark}

\begin{remark}
  One has to be careful when using quantum commitments in a larger system if the receiver's state is quantum after the commit phase.
  As an example, suppose we design a protocol where the quantum commitment held by the receiver before the reveal phase is used inside another cryptographic protocol.
  Then we might not be able to invoke binding if the state is destroyed, whereas classically the state could always be copied.
  Nevertheless, this is a generic caveat of quantum commitments and is not an artifact of any specific definition of binding.
\end{remark}

\paragraph{Computational Hiding.} We define the computational hiding property below. This is the natural quantum analogue of the classical computational hiding property. In the literature, this property is also sometimes referred to as quantum concealing. 

\begin{definition}[Computational Hiding]
\label{def:comphiding}
We say that a quantum commitment scheme $(C,R)$ satisfies computational hiding if for any malicious QPT receiver $\{R^*_{\secparam}\}_{\secparam \in \mathbb{N}}$, for any QPT distinguisher $\{D_{\secparam}\}_{\secparam \in \mathbb{N}}$, the following holds:  
    \begin{align*}
    &\bigg| \Pr \left[  D_{\secparam}(\sigma_{R^*}) = 1\ : \, \sigma_{CR^*} \leftarrow \commit \langle C_{\secparam}(0), R^*_{\secparam} \rangle \right] \\
    & \qquad \qquad \qquad - \Pr \left[ D_{\secparam}(\sigma_{R^*}) = 1 \ :\ \sigma_{CR^*} \leftarrow \commit  \langle C_{\secparam}(1), R^*_{\secparam} \rangle \right] \bigg| \leq \nu(\lambda),
    \end{align*}
    for some negligible function $\nu(\cdot)$, where $\sigma_{R^*}$ is obtained by tracing out the committer's part of the state $\sigma_{CR^*}$.
\end{definition}

\subsection{Construction}
\label{sec:comm:cons}

We now present the main theorem of this section, which shows that statistically binding quantum commitment schemes can be constructed from PRFS.

\begin{theorem}
\label{thm:comm:prfs}
Assuming the existence of $(d(\secparam),n(\secparam))$-PRFS satisfying recognizable abort (\Cref{def:classicalgen}) with $2^d \cdot n \ge \comlen$, there exists a commitment scheme satisfying statistical completeness, statistical binding (\Cref{def:stat:binding}) and computational hiding (\Cref{def:comphiding}).
\end{theorem}

\noindent We note that, combined with \Cref{thm:prfs-from-prs} which constructs PRFS generators with $\Omega(\log\lambda)$ input length and recognizable abort property from PRS generators, we can obtain quantum commitment schemes from PRS generators. We present the construction, which is inspired by Naor's commitment scheme~\cite{Naor91}.

The main building block is a $(d(\secparam),n(\secparam))$-PRFS, denoted by  $G = \{G_{\secparam}(\cdot,\cdot)\}_{\secparam \in \mathbb{N}}$.
Since $n \ge 1$, we assume $d(\secparam) = \lceil\log\frac{\comlen}n\rceil = O(\log\secparam)$ to ensure the efficiency of the algorithm.
This is without loss of generality since we can generically shrink the input length for a PRFS by padding zeroes.
Let $\Test_{\secparam}^{\otimes 2^{d(\secparam)}}$ be the product PRFS tester corresponding to $G$ as guaranteed by \Cref{lem:test-product}. 

We describe the commitment scheme, $(C,R)$ as follows. For notational convenience, we abbreviate $n = n(\lambda)$, $d = d(\lambda)$.

\begin{enumerate}
    \item {\em Commit Phase}: 
    \begin{itemize}
        \item The receiver $R_\secparam$ samples a uniformly random $m$-qubit Pauli operator $P$, where $m = 2^d \cdot n$. We write $P$ as $P_0 \otimes \cdots \otimes P_{2^d - 1}$, where $P_i$ is an $n$-qubit Pauli operator\footnote{To sample $P = \bigotimes_i P_i$, the receiver can sample uniformly random bits $\alpha_1,\beta_1,\ldots,\alpha_m,\beta_m$, and let $P_i = X^{\alpha_i} Z^{\beta_i}$ where $X$ and $Z$ are the single-qubit Pauli operators.}.
        It sends $P$ to the committer. 
        \item The committer $C_\secparam$ on input a bit $b \in \{0,1\}$, does the following: 
        \begin{itemize}
            \item It samples $k \xleftarrow{\$} \{0,1\}^{\secparam}$.
            \item For every $x \in \{0,1\}^{d}$, computes  
        $\sigma_{k,x} \leftarrow G_{\secparam}(k,x)$.
        \end{itemize}
         It sends the commitment $\ciphertext = \bigotimes_{x \in \{0,1\}^{d}} \widetilde{\sigma}_{k,x}$, where $\widetilde{\sigma}_{k,x} = P_{x}^b \sigma_{k,x} P_x^b$, to the receiver. %
    \end{itemize}
    \item {\em Reveal Phase}: The committer sends $(k,b) \in \bit^\lambda \times \bit$ as the decommitment to the receiver.
    The receiver outputs $b$ if and only if $\Test_{\secparam}^{\otimes 2^{d}} \left(\{k,x\}_x, P^b \ciphertext P^b \right)=1$
    where $P^b = \bigotimes_{x \in \{0,1\}^{2^d}} P_x^b$.
    Otherwise the receiver outputs $\bot$.

\end{enumerate}

\begin{lemma}
  If $G$ is a PRFS, then there exists a negligible function $\nu(\cdot)$ such that the probability that the honest receiver accepts the honest committer's opening is at least $1 - \nu(\lambda)$.
\end{lemma}
\begin{proof}
  This follows immediately from \Cref{lem:test-honest} and union bound as $2^d$ is polynomial in $\secparam$.
\end{proof}

\ifcrypto
\else
  \input{commitment-security}

\fi

\input{ot}

%% file: commitment-security.tex
\begin{lemma}
  If $G$ is a PRFS, then $(C,R)$ satisfies computational hiding as defined in \Cref{def:comphiding}. %
\end{lemma}
\begin{proof}
This follows from a standard hybrid argument. Let $R^*$ be a QPT receiver. \\

\paragraph{Hybrid $\hybrid_1$.} This corresponds to $C$ committing to the bit $b=0$. In more details, let $P=\bigotimes_{x \in \{0,1\}^{d}} P_x$. be the Pauli sent by $R^*$ to $C$. Then, $C$ computes $\sigma_{k,x} \leftarrow G_{\secparam}(k,x)$, for every $x \in \{0,1\}^{d}$. $C$ sends $\ciphertext = \bigotimes_{x \in \{0,1\}^{d(\secparam)}} \widetilde{\sigma}_{k,x}$, where $\widetilde{\sigma}_{k,x} = \sigma_{k,x}$ to $R^*$. \\

\paragraph{Hybrid $\hybrid_2$.} This hybrid is the same as before, except that $\sigma_{k,x}=\ketbra{\vartheta_x}{\vartheta_x}$, where $\ket{\vartheta_1},\ldots,\ket{\vartheta_{2^d}} \leftarrow \Haar_n$. 
\par The output distributions of $\hybrid_1$ and $\hybrid_2$ are computationally indistinguishable from the security of PRFS $\{G_{\secparam}\left( \cdot,\cdot \right)\}_{\secparam \in \mathbb{N}}$. \\

\paragraph{Hybrid $\hybrid_3$.} This hybrid is the same as before, except that $\widetilde{\sigma}_{k,x}=P_x \sigma_{k,x}P_x$. That is, the operator $P^b$ is applied to $\sigma_{k,x}$.  
\par The output distributions of $\hybrid_2$ and $\hybrid_3$ are identical by unitary invariance of Haar random states.  \\

\paragraph{Hybrid $\hybrid_4$.} This corresponds to $C$ committing to the bit $b=1$. In more detail, let $P$ be the Pauli sent by $R^*$ to $C$. Then, $C$ computes $\sigma_{k,x} \leftarrow G_{\secparam}(k,x)$, for every $x \in \{0,1\}^{d}$. $C$ sends $\ciphertext = \bigotimes_{x \in \{0,1\}^{d(\secparam)}} \widetilde{\sigma}_{k,x}$, where $\widetilde{\sigma}_{k,x} = P_x \sigma_{k,x} P_x$ to $R^*$.
\par The output distributions of $\hybrid_3$ and $\hybrid_4$ are computationally indistinguishable from the security of PRFS $G$.
\end{proof}

\paragraph{Statistical binding.}
The rest of the section will be devoted to proving statistical binding of the construction.
For this part, we explicitly assume recognizable abort to simplify the analysis.
However, we believe that with some more work, our construction would still satisfy binding even if a more generic PRFS is used.

\begin{lemma}
\label{lem:statbind}
  $(C,R)$ satisfies $O(2^{-\lambda})$-statistical binding if the $(d, n)$-PRFS satisfies recognizable abort property and $2^n \cdot d \ge \comlen$.
\end{lemma}
Let $C^*=\{C^*_{\secparam}\}_{\secparam \in \mathbb{N}}$ be an malicious committer. Suppose $C^*_{\secparam}$ executes the commit phase with the honest receiver $R_{\secparam}$. Let $\ciphertext$ denote the mixed state sent by $C^*$ to $R$.

We first describe the extractor.

\paragraph{Description of ${\cal E}$.} On input the commitment $\ciphertext$, the extractor $\cal E$ obtains the description of the Pauli matrix $P$ from the receiver's state, and performs general measurement $\Lambda$ whose operators are  $\{\sqrt{\Lambda_0},\sqrt{\Lambda_1},\sqrt{\Lambda_{\bot}}\}$, where $\Lambda_0,\Lambda_1,\Lambda_{\bot}$ are positive semi-definite operators defined as follows:
\begin{itemize}
    \item Define $T_0$ to be the subspace spanned by $\left\{ \bigotimes_{x \in \{0,1\}^{2^{d}}}\ketbra{\psi_{k,x}}{\psi_{k,x}}\ :\ \forall k \in \{0,1\}^{\secparam} \right\}$, where the states $\ket{\psi_{k, x}}$ are pure states guaranteed by~\Cref{def:classicalgen}. Let $\Pi_0$ be a projection that projects onto $T_0$. 
    \item Define $T_1$ to be the subspace spanned by $\left\{ P \bigotimes_{x \in \{0,1\}^{2^d}}\ketbra{\psi_{k,x}}{\psi_{k,x}}P\ :\ \forall k \in \{0,1\}^{\secparam} \right\}$. Let $\Pi_1$ be a projection that projects onto $T_1$.  Note that $\Pi_1=P \Pi_0 P$ by definition.
    \item Let $p = \operatornorm{\Pi_0 + \Pi_1}$ (i.e. the maximum eigenvalue of $\Pi_0 + \Pi_1$), $\Lambda_0 = p^{-1} \cdot \Pi_0$, and $\Lambda_1= p^{-1} \cdot \Pi_1$. Define $\Lambda_{\bot} = I - (\Lambda_0 + \Lambda_1)$.
    Since $\Pi_0$ and $\Pi_1$ are projections, $\sqrt{\Lambda_0}$ and $\sqrt{\Lambda_1}$ are well defined. By definition, $\Lambda_0 + \Lambda_1 = p^{-1} (\Pi_0 + \Pi_1) \preccurlyeq I$ and therefore $\Lambda_\bot$ is positive-semidefinite. Thus  $\sqrt{\Lambda_{\bot}}$ is also well-defined. 
\end{itemize}

Let the measurement outcome be $b' \in \{0,1,\bot\}$ and let the post-measurement state be denoted by $\ciphertext'$ after applying the general measurement $\{\sqrt{\Lambda_0}, \sqrt{\Lambda_1}, \sqrt{\Lambda_\bot}\}$. The extractor $\cal E$ outputs $(\ciphertext',b')$. This completes the description of the extractor.

\begin{fact}
\label{clm:paulioverlap}
Let $\ket{\phi}$ and $\ket{\psi}$ be two arbitrary $m$-qubit states. Let ${\cal P}_m$ be the $m$-qubit Pauli group. Then,
$$
\E_{P \leftarrow \cal{P}_m} \left [ \left|\bra{\psi} P \ket{\phi} \right|^2 \right ] = 2^{-m}.
$$
\end{fact}
\begin{proof}
We first observe that $\left|\bra{\psi} P \ket{\phi} \right|^2 = \Tr\left( P \ketbra{\psi}{\psi} P \ketbra{\phi}{\phi} \right)$; this follows from the fact that the trace of an outer product of two vectors is equivalent to the square of their inner product.  

We also use the following fact from~\cite{MTdW00}: for any $m$-qubit density matrix $\rho$,
\begin{eqnarray}
\label{eqn:qotp}
\mathbb{E}_{P \leftarrow {\cal P}_m} \left[  P \rho P \right] = \frac{I}{2^m}~.
\end{eqnarray}
This implies that for all states $\ket{\psi}, \ket{\phi}$,
\begin{align*}
\mathbb{E}_{P \leftarrow {\cal P}_m} \left[ \left|\bra{\psi} P \ket{\phi} \right|^2  \right]  & =  \mathbb{E}_{P \leftarrow {\cal P}_m} \left[ \Tr\left( P \ketbra{\psi}{\psi} P \ketbra{\phi}{\phi} \right)  \right] \\
 & = \Tr \left( \mathbb{E}_{P \leftarrow {\cal P}_m} \left[ P \ketbra{\psi}{\psi} P \ketbra{\phi}{\phi}  \right] \right) & \text{(from linearity of }\mathbb{E})\\
& = \Tr \left( \mathbb{E}_{P \leftarrow {\cal P}_m} \left[ P \ketbra{\psi}{\psi} P \right] \cdot \ketbra{\phi}{\phi} \right) \\
& = \Tr \left( \frac{I}{2^m} \cdot \ketbra{\phi}{\phi} \right) &  \text{(from }~\eqref{eqn:qotp}) \\
& = \frac{1}{2^m} \cdot \Tr \left( \ketbra{\phi}{\phi} \right) \\
&= \frac{1}{2^m}
\end{align*}
as desired.
\end{proof}

\begin{lemma}[Almost orthogonality of $\Pi_0$ and $\Pi_1$]
\label{clm:povm:closeness}
$$\Pr_{P \leftarrow {\cal P}_{m}}\left[ p \geq 1 + 3\cdot 2^{-(m-4\lambda)/3} \right] \leq 2^{-(m-4\lambda)/3}.$$
\end{lemma}
\begin{proof}
Let $\ket{\psi}$ be an arbitrary $m$-qubit pure state.
Write $\ket{\psi} = \ket{\alpha} + \ket{\beta}$, where $\ket{\alpha}$ is the projection of $\ket{\psi}$ onto the subspace $T_0$, and $\ket{\beta}$ is the projection of $\ket{\psi}$ onto the orthogonal complement of $T_0$.
We determine an upper bound for the following quantity: 
\begin{align*}
\bra{\psi} (\Pi_0 + \Pi_1) \ket{\psi} & = \left( \bra{\alpha} + \bra{\beta} \right)(\Pi_0 + \Pi_1) \left( \ket{\alpha} + \ket{\beta} \right) \\
&= \left( \bra{\alpha} + \bra{\beta} \right)\left( \ket{\alpha} + \Pi_1\ket\alpha + \Pi_1\ket{\beta} \right) \\
&= \braket{\alpha|\alpha} + \bra{\alpha} \Pi_1 \ket{\beta} + \bra{\beta} \Pi_1 \ket{\alpha} + \bra{\beta} \Pi_1 \ket{\beta} + \braket{\alpha|\Pi_1|\alpha} \\
&\leq \ip{\alpha}{\alpha} + \ip{\beta}{\beta} + 2 \left |\bra{\alpha} \Pi_1 \ket{\beta} \right | + \braket{\alpha|\Pi_1|\alpha} \\
&= 1 + 2 \left |\bra{\alpha} \Pi_1 \ket{\beta} \right | + \braket{\alpha|\Pi_1|\alpha} \\
&= 1 + 2 \sqrt{\braket{\alpha|\Pi_1|\beta}\braket{\beta|\Pi_1|\alpha}} + \braket{\alpha|\Pi_1|\alpha} \\
&\le 1 + 2\sqrt{\braket{\alpha|\Pi_1|\alpha}} + \braket{\alpha|\Pi_1|\alpha} \\
&\le 1 + 3\sqrt{\braket{\alpha|\Pi_1|\alpha}} \\
&\leq 1 + 3\sqrt{\Tr(\Pi_0 \Pi_1)}
\end{align*}
where we used the fact that since $\ket{\alpha}$ is contained in the support of $\Pi_0$, we have $\ketbra{\alpha}{\alpha} \preceq \Pi_0$ and thus $\Tr(\ketbra{\alpha}{\alpha} \, \Pi_1) \leq \Tr(\Pi_0 \, \Pi_1)$. 

We now estimate the quantity $\Tr(\Pi_0 \, \Pi_1)$. Let $\{\ket{u_1},\ldots,\ket{u_{\dim(T_0)}} \}$ be an orthonormal basis of $T_0$, so that $\Pi_0 = \sum_{i=1}^{\dim(T_0)} \ketbra{u_i}{u_i}$. Using that $\Pi_1 = P \Pi_0 P$ and cyclicity of trace, we have
\[
    \Tr(\Pi_0 \, \Pi_1) = \sum_{i,j=1}^{\dim(T_0)} \bra{u_j} P \ketbra{u_i}{u_i} P \ket{u_j} \leq 2^{2\lambda} \cdot \max_{i,j} |\bra{u_i} P \ket{u_j}|^2
\]
where we used that $\dim(T_0) \leq 2^\secparam$. 

Now, applying \Cref{clm:paulioverlap} to $|\bra{u_i} P \ket{u_j}|^2$ and using Markov's inequality we get that for each $i,j \in [\dim(T_0)]$ we have for all $\delta > 0$,
\[
    \Pr_{P \leftarrow \cal{P}_m} \left [ |\bra{u_i} P \ket{u_j}|^2 \geq \delta \right ] \leq \delta^{-1} 2^{-m}.
\]
Using a union bound over all $i,j$, 
\[
    \Pr_{P \leftarrow \cal{P}_m} \left [\exists \, i,j  : |\bra{u_i} P \ket{u_j}|^2 \geq \delta  \right ] \leq \delta^{-1} 2^{2\secparam - m}
\]
which implies
\[
    \Pr_{P \leftarrow \cal{P}_m} \left [ \Tr(\Pi_0 \, \Pi_1) \geq  \delta \, 2^{2\secparam} \right ] \leq \delta^{-1} \, 2^{2\secparam - m}~.
\]
Putting everything together, since for all $\ket{\psi}$ the quantity $\bra{\psi} (\Pi_0 + \Pi_1) \ket{\psi}$ is upper-bounded by a quantity that only depends on $\Tr(\Pi_0 \Pi_1)$ which only depends on $P$, we get 
$$\Pr_{P \leftarrow {\cal P}_{m}}\left[ \max_{\ket{\psi}} \{ \bra{\psi} (\Pi_0 + \Pi_1) \ket{\psi} \}  \geq 1 + 3 \sqrt{\delta} \, 2^{\secparam} \right] \leq
\delta^{-1} \, 2^{2 \secparam - m}~.
$$
Setting $\delta = 2^{2(\secparam - m)/3}$ we get the desired lemma statement.
\end{proof}

\paragraph{Indistinguishability of Real World and Ideal World.} We need to show that the output distributions of $\realexpt_{\secparam}^{C^*}$ and $\idealexpt_{\secparam}^{C^*,{\cal E}}$ as defined in \Cref{def:stat:binding} are statistically indistinguishable.\\

\noindent To argue this, we set up some notation. 
\begin{itemize}
    \item We assume that after the commit phase, the random Pauli $P$ sent by $R$ in the first message and $C^*$'s decommitment $(k,b)$ are obtained by measuring some registers of their joint state. 
    Let $\sigma_{\reg{X} \reg{Y}}$ denote the joint state of $C^*$ and $R$ after the commit phase, conditioned on the Pauli $P$ and the decommitment $(k,b)$. The register $\reg{X}$ denotes $C^*$'s private register and $\reg{Y}$ denotes $R$'s private register. 
    
    \item Let $\rho_{{\sf real}}$ denote the output of $\realexpt^{C^*}_\secparam$. If $b = \bot$, then  $\rho_{{\sf real}} = (\sigma_{\reg{X}},\ketbra{\bot}{\bot})$. Otherwise, since the $\Test_\secparam^{\otimes 2^d}$ is being applied to register $\reg{Y}$ of $P^b \sigma_{\reg{X} \reg{Y}} P^b$ (where $P^b$ is applied to register $\reg{X}$), we have
    \begin{align*}
        \rho_{{\sf real}} &= \E_{P,k,b} \Tr_{\reg{Y}} \Big( \Test_\secparam^{\otimes 2^d}( \{k,x\}_x , P^b \sigma_{\reg{X} \reg{Y}} P^b) \Big ) \\
        &= \E_{P,k,b} \Tr_{{\reg Y}}\left( M_0 P^b \sigma P^b \right) \otimes \ketbra{b}{b}  + \Tr_{{\reg Y}}\left( M_\bot P^b \sigma P^b \right) \otimes \ketbra{\bot}{\bot}
    \end{align*}
    where $M_0 = \eta^2 \ketbra{\psi}{\psi}$ and $M_\bot = I - M_0$ are positive semi-definite operators acting on register $\reg{Y}$ with $\eta, \ket{\psi}$ given by  \Cref{lem:test-product}. The expectation is over the choice of random Pauli $P$ and decommitment $(k,b)$.
    
    \item Let $\rho_{{\sf ideal}}$ denote the output of $\idealexpt^{C^*,\cal E}_\secparam$. If $b = \bot$, then by definition of the ideal experiment, the output is $(\sigma_{\reg{X}},\ketbra{\bot}{\bot})$. Otherwise, the general measurement $\{ \sqrt{\Lambda_0}, \sqrt{\Lambda_1}, \sqrt{\Lambda_\bot} \}$ is performed first on register $\reg{Y}$ of the state $\sigma_{\reg{X} \reg{Y}}$ to yield outcome $a \in \{0,1,\bot\}$. Conditioned on outcome $a$ the post-measurement state is $\frac{\sqrt{\Lambda_a} \sigma \sqrt{\Lambda_a}}{\Tr(\Lambda_a \sigma)}$.
    
    The Pauli operator $P^b$ and then the $\Test^{\otimes 2^d}_\secparam$ circuit is applied to register $\reg{Y}$ (corresponding to the reveal phase); if the test accepts and the decommitted bit $b$ matches the output $a$ of the extractor, then the register $\reg{X}$ and $\ketbra{\mu}{\mu}$ are output. If the bits do not match then the outcome is $\ketbra{\mathfrak E}{\mathfrak E}$. Otherwise the register $\reg{X}$ and $\ketbra{\bot}{\bot}$ are output. Put together, we get
    \begin{align*}
    \rho_{{\sf ideal}} &= \E_{P,k,b} \left[\Tr_{\reg{Y}} ( N_b \sigma) \otimes \ketbra{b}{b} + \Tr_{\reg{Y}} ( N_\perp \sigma) \otimes \ketbra{\bot}{\bot} + \Tr_{\reg{Y}} (N_{\mathfrak E} \sigma) \otimes \ketbra{\mathfrak E}{\mathfrak E}\right]
    \end{align*}
    where $N_b = \sqrt{\Lambda_b} P^b M_0 P^b \sqrt{\Lambda_b}$, $N_{\mathfrak E} = \sqrt{\Lambda_{1 - b}} P^b M_0 P^b \sqrt{\Lambda_{1 - b}}$, and $N_\bot = I - N_b - N_{\mathfrak E}$. To see that this is correct in the case that the ideal experiment where the receiver accepts, consider that the post-measurement state of the extractor measurement, conditioned on obtaining outcome $b$, is $\frac{\sqrt{\Lambda_b} \sigma \sqrt{\Lambda_b}}{\Tr(\Lambda_b \sigma)}$. Applying $P^b$, conditioning on $\Test_\secparam^{\otimes 2^d}$ accepting, and then tracing out the register $\reg{Y}$ yields the state
    \[
        \Tr_{\reg{Y}} \Big( M_0 \Big( P^b \sqrt{\Lambda_b} \sigma \sqrt{\Lambda_b} P^b \Big) \Big).
    \]
    Note that all the operators $M_0, P^b, \sqrt{\Lambda_b}$ all act on the register $\reg{Y}$, and the partial trace over $\reg{Y}$ is cyclic with respect to such operators. Thus this is equal to $\Tr_{\reg{Y}}(N_b \sigma)$.

\end{itemize}

\noindent We now prove~\Cref{lem:statbind}. %

\begin{proof}[Proof of \Cref{lem:statbind}]
Write 
\begin{align*}
    \rho_{{\sf real}} &= \E_{P,k,b} \left[\tau_{{\sf real}}^{(b)} \otimes \ketbra{b}{b} + \tau_{{\sf real}}^{(\bot)} \otimes \ketbra{\bot}{\bot}\right] \\
    \rho_{{\sf ideal}} &= \E_{P,k,b} \left[\tau_{{\sf ideal}}^{(b)} \otimes \ketbra{b}{b} + \tau_{{\sf ideal}}^{(\bot)} \otimes \ketbra{\bot}{\bot} + \tau^{(\mathfrak{E})}_{{\sf ideal}} \otimes \ketbra{\mathfrak E}{\mathfrak E}\right]
\end{align*}
for subnormalized density matrices $\tau_{{\sf real}}^{(\cdot)},\tau_{{\sf ideal}}^{(\cdot)}, \tau^{(\mathfrak{E})}_{{\sf ideal}} = \Tr_{\reg{Y}} (N_{\mathfrak E} \sigma)$ %
which implicitly depend on $P,k,b$. Since the trace distance is jointly convex we have
\begin{align*}
    &\equad \TraceDist{\rho_{{\sf real}}}{\rho_{{\sf ideal}}} \\
    &\leq \E_{P,k,b} \TraceDist{\tau_{{\sf real}}^{(b)} \otimes \ketbra{b}{b} + \tau_{{\sf real}}^{(\bot)} \otimes \ketbra{\bot}{\bot}}{\tau_{{\sf ideal}}^{(b)} \otimes \ketbra{b}{b} + \tau_{{\sf ideal}}^{(\bot)} \otimes \ketbra{\bot}{\bot} + \tau^{(\mathfrak{E})}_{{\sf ideal}} \otimes \ketbra{\mathfrak E}{\mathfrak E}}\\
    &= \E_{P,k,b}\left[\TraceDist{\tau_{{\sf real}}^{(b)}}{\tau_{{\sf ideal}}^{(b)}} + \TraceDist{\tau_{{\sf real}}^{(\bot)}}{\tau_{{\sf ideal}}^{(\bot)}} + \Tr(\tau^{(\mathfrak{E})}_{{\sf ideal}}) \right].
\end{align*}
Define $\mu = \Tr(\tau^{(\mathfrak{E})}_{{\sf ideal}})$. Using that $M_\bot = I - M_0$ and $N_\bot = I - N_b - N_{\mathfrak E}$ and that the partial trace is cyclic with respect to operators acting on $\reg{Y}$ only, we have
\begin{align*}
\TraceDist{\tau_{{\sf real}}^{(\bot)}}{\tau_{{\sf ideal}}^{(\bot)}} &= \TraceDist{ \Tr_{\reg{Y}} \Big( M_\bot P^b \sigma P^b \Big)}{ \Tr_{\reg{Y}} \Big( N_\bot \sigma \Big) } \\
&= \TraceDist{ \Tr_{\reg{Y}} \Big(  P^b M_0 P^b \sigma \Big)}{\Tr_{\reg{Y}} \Big( (N_b + N_{\mathfrak E}) \sigma \Big) }\\
&\le \TraceDist{\tau_{{\sf real}}^{(b)}}{\tau_{{\sf ideal}}^{(b)}} + \mu.
\end{align*}
Thus to prove the Lemma it suffices to prove that the following statement is true: for all $k \in \bit^\secparam$ and $b \in \bit$, 
\begin{equation}
    \label{eq:commit-sufficient}
    \E_{P \leftarrow \cal{P}_m} \left[\TraceDist{\tau^{(b)}_{{\sf real}}}{\tau^{(b)}_{{\sf ideal}}} + \mu\right] \leq \frac{5}{2^\secparam}~.
\end{equation}

Fix a decommitment $(k,b)$. Recall that $M_0 = \eta^2 \ketbra{\psi}{\psi}$ where $\ket{\psi} = \bigotimes_x \ket{\psi_{k,x}}$, and that $\Lambda_b = p^{-1} \cdot \Pi_b$. Then
\begin{align*}
    N_b &= \sqrt{\Lambda_b} P^b M_0 P^b \sqrt{\Lambda_b} \\
    &= \eta^2 \, p^{-1} \, \sqrt{\Pi_b} P^b \ketbra{\psi}{\psi} P^b \sqrt{\Pi_b} \\
    &= \eta^2 \, p^{-1} \, \Pi_b \, P^b  \ketbra{\psi}{\psi} P^b \, \Pi_b
\end{align*}
where we use the fact that $\sqrt{\Pi_b} = \Pi_b$ since it is a projector. Since $\Pi_b$ projects onto the span of $\{ P^b \bigotimes_x \ket{\psi_{k,x}} : k \in \bit^\secparam \}$, this means that $\Pi_b P^b \ket{\psi} = P^b \ket{\psi}$, so $N_b$ is equal to
\[
 \eta^2 \, p^{-1} \, P^b  \ketbra{\psi}{\psi} P^b = p^{-1} P^b M_0 P^b~.
\]
This means that 
\begin{align}
\E_{P \leftarrow \cal{P}_m} \TraceDist{\tau^{(b)}_{{\sf real}}}{\tau^{(b)}_{{\sf ideal}}} &= \E_{P \leftarrow \cal{P}_m} \TraceDist{ \Tr_{\reg{Y}} \Big(  P^b M_0 P^b \sigma \Big)}{\Tr_{\reg{Y}} \Big( N_b \sigma \Big) } \notag \\
&= \E_{P \leftarrow \cal{P}_m} \TraceDist{ \Tr_{\reg{Y}} \Big(  P^b M_0 P^b \sigma \Big)}{p^{-1} \, \Tr_{\reg{Y}} \Big( P^b M_0 P^b \sigma \Big) } \label{eq:commit-almost}
\end{align}
If $p$ (which is a function of $P$) is at most $1 + 3 \cdot 2^{-(4\lambda - m)/3}$ then we say $p$ is \emph{good}, otherwise it is \emph{bad}. By \Cref{clm:povm:closeness} $p$ is bad with probability at most $2^{-(m-4\lambda)/3}$. When $p$ is good, we have 
\[
\TraceDist{ \Tr_{\reg{Y}} \Big(  P^b M_0 P^b \sigma \Big)}{p^{-1} \, \Tr_{\reg{Y}} \Big( P^b M_0 P^b \sigma \Big) } \leq 1 - p^{-1} \leq 3 \cdot 2^{-(4\lambda - m)/3}
\]
where we used that $\TraceDist{\varphi}{p^{-1} \varphi} \leq 1 - p^{-1} \leq p - 1$ for all subnormalized density matrices $\varphi$.

On the other hand, similarly we have
\begin{align*}
  N_{\mathfrak E}
    &= \sqrt{\Lambda_{1 - b}} P^b M_0 P^b \sqrt{\Lambda_{1 - b}} \\
    &= \eta^2 p^{-1} \Pi_{1 - b} P^b \ketbra\psi\psi P^b \Pi_{1 - b} \\
    &= \eta^2 p^{-1} P^{1 - b} \Pi_0 P\ketbra\psi\psi P\Pi_0 P^{1 - b},
\end{align*}
where we use the same facts as before and in addition $\Pi_i = P^i \Pi_0 P^i$ for $i = 0, 1$.
Since $\Tr(\sigma) = 1$, $p \ge 1$ and $\eta^2 \le 1$,
\begin{align*}
  \mu
    &= \Tr(N_{\mathfrak E} \sigma) \\
    &\le \operatornorm{N_{\mathfrak E}} \, \Tr(\sigma) \\
    &= \eta^2 p^{-1} \operatornorm{\Pi_0 P\ketbra\psi\psi P\Pi_0} \\
    &\le \braket{\psi|P\Pi_0 P|\psi}.
\end{align*}
Thus $\E_{P \leftarrow \cal{P}_m} \mu \le 2^{\lambda - m}$ by the dimensionality of $\Pi_0$.

Therefore by \eqref{eq:commit-almost}, LHS of \eqref{eq:commit-sufficient} is at most
\[
    3 \cdot 2^{-(m-4\lambda)/3} + 2^{-(m-4\lambda)/3} + 2^{-(m - \lambda)} \le 4 \cdot 2^{-(m-4\lambda)/3} + 2^{-(m - \lambda)}.
\]
Averaging over $k,b$, we get that $\TraceDist{\rho_{{\sf real}}}{\rho_{{\sf ideal}}} \leq 8 \cdot 2^{-(m-4\lambda)/3} + 2 \cdot 2^{-(m - \lambda)}$, which is less than $\frac{10}{2^\secparam}$ when $m \geq \comlen$ as desired.
\end{proof}

%% file: ot.tex
\subsection{Application: Secure Computation}

In this section, we show how to base secure computation solely on the existence of a PRS.
While there are two works \cite{BartusekCKM21a,GLSV21} showing that post-quantum one-way functions and quantum communication suffice to obtain protocols for secure computation, the construction of Bartusek, Coladangelo, Khurana, and Ma~\cite{BartusekCKM21a} has the advantage that it uses the starting commitment scheme as a black box. %
We recall their main theorem.

\begin{theorem}[Implicit from \cite{BartusekCKM21a}]
\label{thm:mpc:statcomm}
 Assuming the existence of quantum statistically binding bit commitments, maliciously secure computation protocols (in the dishonest majority setting) for $P/poly$ exist.
\end{theorem}

\paragraph{Comparison of the definitions of statistical binding.}
The application of \Cref{thm:mpc:statcomm} would be straightforward except for one subtlety, which is that we are using a more general definition of the statistical binding property than required by their work.
Their notion of statistical binding is tailored to commitment schemes with classical messages as it suffices for their purposes.
We first recall their definition of statistical binding in the full version of their work~\cite{bartusek2021oneway}, and show that it seems strictly stronger than our definition.
\begin{definition}[{\cite[Definition 3.2]{bartusek2021oneway}}]
  \label{def:bckm-binding}
  A bit commitment scheme is statistically binding if for every unbounded-size committer $\mathcal C^*$,
  there exists a negligible function $\nu(\cdot)$ such that with probability at least $1 - \nu(\lambda)$ over the measurement randomness in the commitment phase, there exists a bit $b \in \bit$ such that the probability that the receiver accepts $b$ in the reveal phase is at most $\nu(\lambda)$.
\end{definition}
\begin{lemma}
  \label{lem:binding-generalized}
  If a commitment scheme satisfies \Cref{def:bckm-binding}, then it also satisfies \Cref{def:stat:binding}.
\end{lemma}
\begin{proof}
  Since a malicious committer can always ``purify'' his measurements via the deferred measurement principle, without loss of generality we assume the only measurements in the commit phase are only done by the honest receiver.
  By \Cref{def:bckm-binding}, there exists a classical function $\mathcal E$ that maps the honest receiver's measurement outcomes $m$ to a bit so that the probability that the receiver accepts $1 - \mathcal E(m)$ is negligible (also known as the correctness of the extractor).
  As $\mathcal E$ only acts on the measurement outcome that is therefore guaranteed to be classical, $\mathcal E$ commutes with the committer's and receiver's operations.
  Furthermore, the output of $\mathcal E$ is also classical by definition.
  Therefore, the only difference between the real world and the ideal world is introduced by the extraction error in the ideal world, and thus the statistical indistinguishability follows immediately by the correctness of the extractor.
\end{proof}

Our protocol cannot satisfy this property since the honest receiver does not measure the committer's message in any way, and therefore in general it is possible for the committer to generate an equal superposition of commitment to 0 and commitment to 1, in which case this binding property is violated, as the receiver will open to 0 and 1 with equal probability.
Nonetheless, \Cref{def:stat:binding} is very similar to \Cref{def:bckm-binding}.
In particular, \Cref{def:stat:binding} says that there is an implicit measurement that could be done to extract the committed bit in a way unnoticeable to the malicious committer as well as the honest receiver.
Intuitively, whenever we would like to invoke \Cref{def:bckm-binding}, we can switch to the ideal world where the bit is extracted, and then this ``ideal scheme'' essentially satisfies \Cref{def:bckm-binding}.
We formalize this intuition with the following lemma.

\begin{definition}
  We call $(C, R)$ a quantum commitment scheme with an inefficient receiver if it satisfies the requirements of a commitment scheme except that $R$ need not be a QPT algorithm.
  
  Let $(C, R)$ and $(C, R')$ be two quantum commitment schemes with an inefficient receiver.
  We call them statistically indistinguishable against malicious committers, if the outcome of any (unbounded) nonuniform experiment described below can only distinguish $R$ from $R'$ with negligible advantage.
  \begin{itemize}
    \item The algorithm has an arbitrary non-uniform input state $\ket{\psi_\lambda}$, and interacts as a committer with either $R$ or $R'$ via the commitment scheme.
    \item The algorithm can choose to abort the interaction at any stage. Otherwise at the end of the interaction, $R$ or $R'$ outputs his decision as a classical symbol $\mu \in \{0, 1, \bot\}$ to the algorithm.
    \item The algorithm performs an arbitrary channel on his internal state as the output.
  \end{itemize}
\end{definition}
\begin{lemma}
  \label{lem:binding-partial-converse}
  If a commitment scheme $(C, R)$ satisfies \Cref{def:stat:binding}, then there is a commitment scheme $(C, \tilde R)$ with an inefficient receiver that satisfies \Cref{def:bckm-binding}.
  Furthermore, these two commitment schemes are statistically indistinguishable against malicious committers; and $\tilde R$ is the same as $R$, except that at the end of the commit phase, the extractor $\mathcal E$ of $(C, R)$ is applied on the receiver's state, and its output is saved in a separate register.
\end{lemma}
\begin{proof}
  Note that $(C, \tilde R)$ is the same receiver as the ideal experiment of \Cref{def:stat:binding}, except that at the end we always run the honest receiver as usual instead of checking whether the extraction is correct, and therefore this change is statistically indistinguishable to the committer by \Cref{def:stat:binding}.
  
  To show that it satisfies \Cref{def:bckm-binding}, we notice that assume the extractor's measurement outcome is $b$ (if it is $\bot$ then set $b$ to 0), the probability that the committer can open to $1 - b$ is negligible, as otherwise the ideal world will have a non-negligible weight on extraction error $\ketbra{\mathfrak E}{\mathfrak E}$, which contradicts \Cref{def:stat:binding}.
\end{proof}

It is not hard to see that by leveraging \Cref{lem:binding-generalized,lem:binding-partial-converse}, we can recover \Cref{thm:mpc:statcomm} even with our definition of statistical binding (\Cref{def:stat:binding}).
The proof of this is not very enlightening and we defer the details to \Cref{sec:ot-details}.
By instantiating the statistically binding bit commitments in~\Cref{thm:mpc:statcomm} with PRS (\Cref{thm:comm:prfs} and~\Cref{thm:prfs-from-prs}), we obtain the following corollary.

\begin{corollary}
  Assuming the existence of $(2\log\secparam + \omega(\log\log\secparam))$-PRS, there exists maliciously secure computation protocol for P/poly in the dishonest majority setting. 
\end{corollary}

%% file: cpa.tex
\subsection{CPA-Secure Quantum Encryption Schemes}
\label{sec:cpa}

In \Cref{sec:qotp} we presented the \emph{Quantum Pseudo One-Time Pad}, which is an encryption scheme that has one-time security (meaning that the key can only be used to encrypt one message before losing security). In this section we present an encryption scheme that has \emph{many-time} security. 

More formally, we construct a symmetric-key quantum encryption scheme that is secure against \emph{chosen plaintext attacks}, or in other words satisfies \emph{CPA security}. For simplicity we focus on schemes that encrypt a single bit; in the many-time security setting, this is equivalent to being able to encrypt many bits. We also restrict our attention to a version of CPA security where the adversary makes \emph{classical}, nonadaptive queries to the encryption oracle; we believe that it should be straightforward to adapt the proof to handle adaptive and even quantum queries. \hnote{do we actually think it's straightforward?}

\begin{definition}[CPA-Secure Quantum Encryption Scheme]
\label{def:cpa}
We say that a pair of QPT algorithms $(\Enc,\Dec)$ is a \emph{CPA-secure symmetric-key quantum encryption scheme for bits} if the following properties are satisfied:
\begin{itemize}
    \item \textbf{Correctness}: There exists a negligible function $\eps(\cdot)$ such that for every $\lambda$ and every $b \in \{0,1\}$,
    $$\Pr_{\substack{k \leftarrow \{0,1\}^\lambda,\\ \sigma \leftarrow \Enc_\lambda(k,b)}} \left[ \Dec_\lambda(k,\sigma) = b \right] \geq 1 - \eps(\lambda).$$
    \item
    \textbf{Security}: For every polynomial $t(\cdot)$, for every nonuniform QPT adversary $A$, there exists a negligible function $\eps(\cdot)$ such that for all $\secparam$ the adversary has at most $\eps(\secparam)$ advantage in the following security game:
    \begin{enumerate}
        \item Challenger samples key $k \leftarrow \{0,1\}^\secparam$ and $z \leftarrow \{0,1\}$.
        \item The adversary sends $(b_1^{(0)},b_1^{(1)}),\ldots,(b_t^{(0)},b_t^{(1)})) \in \{0,1\}^{2t}$ and receives $$\left( \Enc_\secparam \left(k,b_1^{(z)} \right),\ldots,\Enc_\secparam \left(k,b_t^{(z)} \right) \right).$$
        \item The adversary outputs $z' \in \{0,1\}$ and wins if $z' = z$.
    \end{enumerate}
    Here we have abbreviated $t = t(\secparam)$.
\end{itemize}
\end{definition}

We now present our CPA-secure quantum encryption scheme $(\Enc,\Dec)$ for bits. Let $G$ be a $(d(\lambda),n(\lambda))$-PRFS generator where $d(\secparam),n(\lambda) = \omega(\log \lambda)$. Let $\Test$ denote the test algorithm from \Cref{lem:test-honest}. Fix $\lambda$ and let $d = d(\lambda)$ and $n = n(\lambda)$.

\begin{enumerate}
    \item $\Enc_\lambda(k,b)$: on input $k \in \{0,1\}^\lambda$ and a message $b \in \{0,1\}$, do the following: 
    \begin{itemize}
        \item Sample $r \leftarrow \{0,1\}^{d-1}$.
        \item Set $\sigma = G(k,(r,b))$.
        \item Output the ciphertext state $(r,\sigma)$.
    \end{itemize}

    \item $\Dec_\lambda(k,r,\sigma)$: on input $k \in \{0,1\}^\secparam$, $r \in \{0,1\}^{d-1}$, and $n$-qubit state $\sigma$, perform the following operations: 
    \begin{itemize}
        \item Execute $\Test(k,(r,0),\sigma)$. If it accepts, set $b = 0$, otherwise, set $b = 1$.
        \item Output $b$.
    \end{itemize}
\end{enumerate}

\begin{lemma}
$(\Enc,\Dec)$ satisfies the correctness property of a CPA-secure quantum encryption scheme according to \Cref{def:cpa}.
\end{lemma}
\begin{proof}
The proof of correctness is virtually identical to that of \Cref{lem:qopt-correctness}.
\end{proof}

\begin{lemma}
$(\Enc,\Dec)$ satisfies the security property of a CPA-secure quantum encryption scheme according to \Cref{def:cpa}.
\end{lemma}
\begin{proof}
Consider the following hybrids.

\paragraph{Hybrid 0.} This is the real security game, where the adversary receives
$$\left( \Enc_\secparam \left(k,b_1^{(z)} \right),\ldots,\Enc_\secparam \left(k,b_t^{(z)} \right) \right),$$ which by definition is 
\[
    \left ( (r_1,G(k,(r_1,b_1^{(z)}))), \ldots, (r_t,G(k,(r_t,b_t^{(z)}))) \right)
\]
where $r_1,\ldots,r_t \in \{0,1\}^{d-1}$ are chosen uniformly at random. 

\paragraph{Hybrid 1.} Instead of sampling $r_1,\ldots,r_t$ independently, they are sampled uniformly at random \emph{conditioned} on them being all distinct. Letting $\mathsf{D}$ denote the event that the $r_1,\ldots,r_t$ are distinct, the total variation distance between the two distributions on $(r_1,\ldots,r_t)$ can be bounded by
\begin{align*}
    &\sum_{(r_1,\ldots,r_t) \in (\{0,1\}^{d-1})^t} \Big | \Pr[r_1,\ldots,r_t] - \Pr[r_1,\ldots,r_t \mid \mathsf{D}] \Big| \\
    & \qquad = \Pr[\neg \mathsf{D}] + \sum_{(r_1,\ldots,r_t) \text{ distinct}} \Big | \Pr[r_1,\ldots,r_t] - \Pr[\mathsf{D}]^{-1} \cdot \Pr[r_1,\ldots,r_t] \Big| \\
    & \qquad = \Pr[\neg \mathsf{D}] + \Pr[ \mathsf{D}]^{-1} \cdot \Pr[\neg \mathsf{D}] \cdot \sum_{(r_1,\ldots,r_t) \text{ distinct}} \Pr[r_1,\ldots,r_t] \\
    & \qquad \leq \Pr[\neg \mathsf{D}] \cdot \Big( 1 + \Pr[\mathsf{D}]^{-1} \Big)~.
\end{align*}
We compute 
\[
    \Pr[\neg \mathsf{D}] \leq \binom{t}{2} 2^{-(d-1)} \leq \frac{t^2}{2^{d-1}}
\]
which is negligible since $d = \omega(\log \secparam)$ and $t = \mathrm{poly}(\secparam)$. Thus the total variation distance between the distribution of inputs received by the adversary in Hybrid 0 versus that of Hybrid 1 is negligible.

This implies the adversary's advantage in Hybrid 1 differs from its advantage in Hybrid 0 by at most a negligible amount.

\paragraph{Hybrid 2.} The adversary instead receives
\[
    \Big ( (r_1,\ket{\vartheta_1}), \ldots, (r_t,\ket{\vartheta_t}) \Big)
\]
where $\ket{\vartheta_1},\ldots,\ket{\vartheta_t}$ are independently sampled Haar-random states. The advantage of the adversary in Hybrid 1 is negligibly different from its advantage in Hybrid 0 because of the pseudorandomness property of the PRFS $G$: since all of the $r_i$'s are distinct, we can replace each of $G(k,(r_i,b_i^{(z)}))$ with an independently sampled Haar-random state $\ket{\vartheta_i}$ (and the advantage only changes negligibly with the replacement).

However, this input is independent of the choice of bits $(b_i^{(0)},b_i^{(1)})_i$ and the hidden bit $z$. Thus the advantage of the adversary in guessing $z$ is $0$. This implies that the adversary's advantage in the real security game (i.e. Hybrid 0) is negligible. This concludes the security proof.
\end{proof}

As mentioned in the introduction, for this construction we require PRFS generators with $\omega(\log \secparam)$ input length, which we do not know how to construct from PRS generators alone. We leave this as an open question.

%% file: macs.tex
\subsection{Message Authentication Codes}
\label{sec:macs}

In this section we present a construction of reusable message authentication codes (MACs) from PRFS. Just like in \Cref{sec:cpa}, for the sake of simplicity we focus on a version of MAC security where the adversary makes classical, nonadaptive queries to the MAC signing oracle. 

\newcommand{\Sign}{\mathsf{Sign}}
\newcommand{\Ver}{\mathsf{Ver}}

\begin{definition}
\label{def:macs}
We say that a pair of QPT algorithms $(\Sign,\Ver)$ is a \emph{reusable quantum message authentication code (MAC) scheme} for messages of length $\ell(\secparam)$ if the following properties are satisfied:
\begin{itemize}
    \item \textbf{Correctness}: There exists a negligible function $\eps(\cdot)$ such that for every $\secparam$ and every $m \in \{0,1\}^\ell$, 
    $$\Pr_{\substack{k \leftarrow \{0,1\}^\lambda,\\ \sigma \leftarrow \Sign_\lambda(k,m)}} \left[ \Ver_\lambda(k,m,\sigma) = 1 \right] \geq 1 - \eps(\lambda).$$    
    \item \textbf{Security}: For every polynomial $t(\cdot)$ and every nonuniform QPT adversary $A$, there exists a negligible function $\eps(\cdot)$ such that for every $\secparam$ the adversary wins with probability at most $\eps(\secparam)$ in the following security game:
    \begin{enumerate}
        \item Challenger samples key $k \leftarrow \{0,1\}^\secparam$.
        \item The adversary sends $m_1,\ldots,m_t \in \{0,1\}^\ell$ and receives $\Sign(k,m_i)$ for $i = 1,\ldots,t$. 
        \item The adversary outputs a pair $(m^*,\sigma)$ where $m^* \notin \{m_1,\ldots,m_t\}$ and $\sigma$ is a quantum state.
        \item The adversary wins if $\Ver(k,m^*,\sigma) = 1$.
    \end{enumerate}
    Here we have abbreviated $t = t(\secparam)$ and $\ell = \ell(\secparam)$.
\end{itemize}
\end{definition}

We now present our MAC scheme $(\Sign,\Ver)$ for messages of length $\ell(\cdot)$. Let $G$ be a $(d(\lambda),n(\lambda))$-PRFS generator where $d(\secparam) \geq \ell(\secparam)$ and $n(\lambda) = \omega(\log \lambda)$. For simplicity, assume that $G$ has perfect state generation. Let $\Test$ denote the test algorithm from \Cref{lem:test-honest}. Fix $\lambda$ and let $d = d(\lambda)$, $\ell = \ell(\secparam)$, and $n = n(\lambda)$.

\begin{enumerate}
    \item $\Sign_\lambda(k,m)$: on input $k \in \{0,1\}^\lambda$ and a message $m \in \{0,1\}^\ell$, do the following: 
    \begin{itemize}
        \item Output $\sigma = G(k,m)$ (where we pad $m$ with trailing zeroes if $\ell < d$).
    \end{itemize}

    \item $\Ver_\lambda(k,m,\sigma)$: on input $k \in \{0,1\}^\secparam$, $m \in \{0,1\}^{\ell}$, and $n$-qubit state $\sigma$, perform the following operations: 
    \begin{itemize}
        \item Execute $\Test(k,m,\sigma)$ (where we pad $m$ with zeroes if $\ell < d$), and output $1$ if it accepts; otherwise, output $0$.
    \end{itemize}
\end{enumerate}

\begin{lemma}
$(\Sign,\Ver)$ satisfies the correctness property of a reusable quantum message authentication code (MAC) scheme according to \Cref{def:macs}.
\end{lemma}
\begin{proof}
The proof of correctness is virtually identical to that of \Cref{lem:qopt-correctness}.
\end{proof}

\begin{lemma}
$(\Sign,\Ver)$ satisfies the security property of a reusable quantum message authentication code (MAC) scheme according to \Cref{def:macs}.
\end{lemma}

\begin{proof}
 Suppose there was an adversary $A$ and polynomial $t(\cdot)$ that won the corresponding MAC security game with non-negligible probability $p(\cdot)$. Then for infinitely many $\secparam$ there exists a sequence of messages $m_1,\ldots,m_t \in \{0,1\}^\ell$ and $m^* \notin \{m_1,\ldots,m_t\}$ where $t = t(\secparam)$ and $\ell = \ell(\secparam)$ such that 
\[
    \Pr_{\substack{k \leftarrow \{0,1\}^\secparam \\ \sigma_i = \Enc_\secparam(k,m_i) \\ \sigma^* \leftarrow A(m^*,(m_i,\sigma_i)_i)}} \left [ \Ver_\secparam(k,m^*,\sigma^*) = 1 \right] \geq p(\secparam)~.
\]
By definition of the verification procedure, we have that for all $k,m^*,\sigma^*$, 
\[
    \Pr \left [ \Ver_\secparam(k,m^*,\sigma^*) = 1 \right] = \bra{\psi_{k,m^*}} \sigma^* \ket{\psi_{k,m^*}}
\]
where, since we assume that the PRFS generator $G$ has perfect state generation, the output of $G(k,m^*)$ is some pure state $\ket{\psi_{k,m^*}}$. 

Consider the following QPT adversary $B$: it gets inputs 
\[
(m_1,\ldots,m_t,m^*,\ket{\psi_{k,m_1}},\ldots,\ket{\psi_{k,m_t}},\ket{\psi_{k,m^*}}),
\]
first runs the adversary $A$ on input $(m^*,(m_i,\ket{\psi_{k,m_i}}))$ to obtain a state $\sigma^*$, and then performs the SWAP test between its copy of $\ket{\psi_{k,m^*}}$ and $\sigma^*$. The acceptance probability is $\frac{1}{2} + \frac{1}{2} \bra{\psi_{k,m^*}} \sigma^* \ket{\psi_{k,m^*}} \geq \frac{1}{2} + \frac{1}{2} p(\secparam)$.

On the other hand, the pseudorandomness property of the PRFS generator implies that the acceptance probability of $B$ is negligibly close to the case when it receives inputs
\[
(m_1,\ldots,m_t,m^*,\ket{\vartheta_1},\ldots,\ket{\vartheta_t},\ket{\vartheta_{t+1}}),
\]
where $\ket{\vartheta_1},\ldots,\ket{\vartheta_{t+1}}$ are independent $n$-qubit Haar-random states (here we assume without loss of generality that $m_1,\ldots,m_t,m^*$ are all distinct). However the acceptance probability in this case is
\[
\frac{1}{2} + \E_{\ket{\vartheta_{t+1}} \leftarrow \Haar(n)} \frac{1}{2} \Tr(\sigma^* \, \ketbra{\vartheta_{t+1}}{\vartheta_{t+1}}) = \frac{1}{2} + \frac{1}{2} \cdot 2^{-n}
\]
which is not negligibly close to $\frac{1}{2} + \frac{1}{2} p(\secparam)$, a contradiction. 
\end{proof}

This MAC scheme requires PRFS generators with input length that is greater than the length of the message being authenticated. Thus we do not know how to construct MAC schemes for long messages using only PRS generators. We leave this is an open question.

%% file: garble.tex
\subsection{Quantum Garbling Schemes for {\sf P/poly} from Quantum Pseudo OTP}

\noindent We construct quantum garbling schemes for {\sf P/poly} from quantum pseudo OTP. We first define the notion of quantum garbling schemes. 

\begin{definition}[Quantum garbling]
A quantum garbling scheme for a class of circuits ${\cal C}$ consists of the following QPT algorithms: 
\begin{itemize}
    \item ${\sf Garble}(1^{\secparam},C)$: it takes as input a security parameter $\secparam$, circuit $C \in {\cal C}$ and outputs a garbled circuit ${\sf GC}$, modeled as a quantum state, along with a secret key $sk$. 
    \item ${\sf InpEncode}(sk,x)$: it takes as input the secret key $sk$, classical input $x$ and output an input encoding $\sigma_x$. 
    \item ${\sf Decode}({\sf GC},\sigma_x)$: it takes as input a garbled circuit ${\sf GC}$, input encoding $\sigma_x$ and outputs a value $\chi$. 
\end{itemize}
\noindent Moreover, $\left({\sf Garble},{\sf InpEncode},{\sf Decode} \right)$ satisfies the following additional properties: 
\begin{itemize}
    \item {\bf Correctness}: for every $C \in {\cal C}$, input $x$, $C(x) \leftarrow {\sf Decode}({\sf GC},\sigma_x)$, where $({\sf GC},sk) \leftarrow {\sf Garble}(1^{\secparam},C)$ and $\sigma_x \leftarrow {\sf InpEncode}(sk,x)$. 
    \item {\bf Security}: There exists a QPT simulator ${\sf Sim}$ such that for every $C \in {\cal C}$, $x \in \{0,1\}^{\ell}$, where $\ell$ is the input of $C$, QPT distinguisher $D$, the following holds: 
    $$ {\rm Pr}\left[ 1 \leftarrow D({\sf GC},\sigma_x)\ :\ ({\sf GC},sk) \leftarrow {\sf Garble}(1^{\secparam},C),\ \sigma_x \leftarrow {\sf InpEncode}(sk,x) \right]$$
    $$ - {\rm Pr} \left[ 1 \leftarrow D({\sf GC},\sigma_x)\ :\ ({\sf GC},\sigma_x) \leftarrow {\sf Sim}(1^{\secparam},C,C(x)) \right] | \leq \nu(\secparam),$$
    for some negligible function $\nu(\cdot)$. 
\end{itemize}
\end{definition}

\noindent Unlike the work of Brakerski and Yuen~\cite{brakerski2020quantum}, who also define quantum garbling schemes, we only consider quantum garbling schemes for classical circuits. 

\noindent \paragraph{Construction.} We start with a quantum pseudo one-time pad scheme, denoted by $({\sf Enc}_{\secparam},{\sf Dec}_{\secparam})$. Specifically, we consider a quantum pseudo one-time pad scheme where the message length is twice the key length. For convenience sake, we omit the subscript $\secparam$ from the algorithms. We construct a quantum garbling scheme for ${\cal C} \in {\sf P/poly}$. We assume, without loss of generality, that every gate has both fan-in and fan-out to be 2.  
\par Our construction is identical to the point-and-permute garbling scheme~\cite{BMR90}, where the length doubling pseudorandom generator is replaced by quantum pseudo one-time pad. We present the construction below. \\

\begin{enumerate}
\item ${\sf Garble}(1^{\secparam},C)$: on input the security parameter $\secparam$, circuit $C$, do the following: 
\begin{itemize}
    \item For every wire $w$ in $C$, we associate two wire labels $k_{w,0} \xleftarrow{\$} \{0,1\}^{2\secparam}$ and $k_{w,1}  \xleftarrow{\$} \{0,1\}^{2\secparam}$. For every $b \in \{0,1\}$, we interpret $k_{w,b}$ to be the concatenation of two $\secparam$-bit strings $k_{w,b}^0$ and $k_{w,b}^1$. 
    \item For every wire $w$ in $C$, we associate a random bit $r_w$. 
    \item For every gate $G$ in $C$, compute a garbled gate consisting of four entries, indexed by $\{0,1\}^2$. Let $w_1,w_2$ be the input wires of $G$ and $w_3,w_4$ be its output wires. The $(b_1,b_2)^{th}$ entry, for $b_1,b_2 \in \{0,1\}$, is of the following form: 
    $$\rho_{G}^{b_1,b_2} = {\sf Enc}\left(k_{w_1}^{b_1} \oplus k_{w_2}^{b_2}, \theta_{G,b_1,b_2} \right) $$
    where $\theta_{G,b_1,b_2}$ is defined as follows: 
    $$\bigg(k_{w_3,G(b_1 \oplus r_{w_1},b_2 \oplus r_{w_2})\oplus r_{w_3}}\ ||\ G(b_1 \oplus r_{w_1},b_2 \oplus r_{w_2})\ ||\ $$
    $$\ ||\ k_{w_4,G(b_1 \oplus r_{w_1},b_2 \oplus r_{w_2}) \oplus r_{w_4}}\ ||\ G(b_1 \oplus r_{w_1},b_2 \oplus r_{w_2}) \oplus r_{w_4} \bigg)$$
    \noindent Denote the concatenation of all ciphertexts to be the garbled table, ${\cal T}_{G}$. 
    \item Let ${\cal W}_{{\sf out}}$ consist of the output wires of $C$. Compute the translation table $\{{\cal O}_{w}\}_{w \in {\cal W}_{\sf out}}$, where ${\cal O}_w$, for every $w \in {\cal W}_{{\sf out}}$, is a mapping of the form: $k_{w,b} \rightarrow b$, for $b \in \{0,1\}$. 
    \item Output the garbled circuit ${\sf GC}=\left(\left\{ {\cal T}_G \right\}_{G \in {\cal C}}, \left\{ {\cal O}_{w} \right\}_{w \in {\cal W}_{{\sf out}}} \right)$ along with the secret key $sk=\{\left( k_{w,b},r_w \right)\}_{b \in \{0,1\},w \in {\cal W}_{{\sf in}}}$, where ${\cal W}_{{\sf in}}$ is the set of input wires of $C$. 
\end{itemize}

\item ${\sf InpEncode}\left(sk,x \right)$: on input secret key $sk=\{\left( k_{w,b},r_w \right)\}_{w \in {\cal W}_{{\sf in}}}$, input $x$, output the following:
$$\sigma_x = \{\left( k_{w,x_{\pi(w)}},r_w \oplus x_{\pi(w)} \right)\}_{w \in {\cal W}_{{\sf in}}},$$
where $\pi:{\cal W}_{{\sf in}} \rightarrow [|x|]$ is a mapping from the input wires to the bits of the input.  \\

\item ${\sf Decode}\left( {\sf GC},\sigma_x \right)$: on input garbled circuit ${\sf GC}$, input encoding $\sigma_x$, we employ the follow iterative process starting from the gates in the input layer and all the way to the output layer. 
\begin{itemize} 
\item For every $G \in {\cal C}$, with input wires $w_1,w_2$, output wires $w_3,w_4$, do the following: let $(k'_{w_1},r'_{w_1})$ and $(k'_{w_2},r'_{w_2})$ be the wire labels recovered for the gate $G$. Compute $${\sf Dec}\left((k'_{w_1})^{r'_{w_1}} \oplus (k'_{w_2})^{r'_{w_2}},\ \rho_{G}^{r'_{w_1},r'_{w_2}} \right)$$ to recover $(k'_{w_3}||r'_{w_3}||k'_{w_4}||r'_{w_4})$, where $(k'_{w_3}||r'_{w_3})$ represent the wire labels associated with the wire $w_3$ and $(k'_{w_4}||r'_{w_4})$ represent the wire labels associated with the wire $w_4$.
\item Output $\{{\cal O}_w \left( k'_{w}\right)\}_{w \in {\cal W}_{{\sf out}}}$.
\end{itemize}
\end{enumerate}

\noindent We omit the correctness and security proofs since the above scheme is identical to the point-and-permute garbling scheme of~\cite{BMR90}.

%% file: ot-details.tex
In this section, we give more details about instantiating \Cref{thm:mpc:statcomm} with \Cref{def:stat:binding}.
In order to give a self-contained presentation of the proof of \Cref{thm:mpc:statcomm} with \Cref{def:stat:binding} in a succinct way, it will be convenient for us to go back and forth between the two definitions using \Cref{lem:binding-generalized,lem:binding-partial-converse}, although in principle we can prove it without talking about \Cref{def:bckm-binding} at all.
In particular, we are going to only focus on the parts of the proof of \Cref{thm:mpc:statcomm} where the statistical binding property is invoked, and inform the reader how the proof changes if \Cref{def:stat:binding} is instead used.
We refer the readers to the original work \cite{bartusek2021oneway} for the full proof as well as the definitions for the terminologies below that are irrelevant to our discussions.

The proof of \Cref{thm:mpc:statcomm} uses the statistical binding property in only two places.
\begin{enumerate}
    \item A special case of {\cite[Theorem 1]{BartusekCKM21a}}: They show how to compile a statistically binding commitment in a way that preserves the statistical binding property, but in addition satisfies an additional property that is irrelevant to our discussion. %
    \item {\cite[Theorem 2]{BartusekCKM21a}}: They show how to go from a statistical-binding commitment scheme (that is the output of step 1) to a statistical-hiding commitment scheme with another additional property that is also irrelevant to our discussion.%
\end{enumerate}

\begin{definition}
  Let $(C, R)$ be a commitment scheme.
  A commitment scheme $(C', R')$ is called committer-black-box compiled from $(C, R)$, if its commit phase satisfies the following template:
  \begin{itemize}
    \item $C'$ picks some uniformly random bits $r$ and commits to them using $(C, R)$.
    \item $R'$ randomly picks some of the commitments from the previous step to open, and aborts if $R'$ thinks $C'$ is malicious.
    \item $C'$ computes a (randomized) function that only depends on $r$ and the committed bit $b$, and sends the output to $R'$.
  \end{itemize}
  The function mapping $(C, R)$ to $(C', R')$ is called the committer-black-box compiler.
\end{definition}

We can verify that the compiler~\cite[Fig. 1]{BartusekCKM21a} in the first step indeed satisfies this definition by staring at the construction.
Since the statistical binding property is only invoked to prove the compiled scheme also satisfies statistical binding, it suffices to prove the following lemma.
\begin{lemma}
  If committer-black-box compiler preserves \Cref{def:bckm-binding} for any $(C, R)$, then it also preserves \Cref{def:stat:binding} for any $(C, R)$.
\end{lemma}
\begin{proof}
  Let $(C, R)$ be a commitment scheme that satisfies \Cref{def:stat:binding}.
  Let $(C, \tilde R)$ be the commitment scheme with an inefficient receiver corresponding to \Cref{lem:binding-partial-converse}.
  We establish that $(C', R')$ satisfies \Cref{def:stat:binding} for any malicious $C'$ via the following sequence of hybrids.
  We one by one replace the invocations of $(C, R)$ with $(C, \tilde R)$.
  Each change is statistically indistinguishable to $C'$ by statistical indistinguishability of $R$ and $\tilde R$.
  Denote the scheme after the change $(C', \bar R)$ (since we only make changes to the honest receiver).
  
  Since $(C, \tilde R)$ satisfies \Cref{def:bckm-binding}, by the premise of the problem, the scheme after the change $(C', \bar R)$, also satisfies \Cref{def:bckm-binding}.
  Invoking \Cref{lem:binding-generalized}, $(C', \bar R)$ also satisfies \Cref{def:stat:binding}.
  This concludes the proof that the real world of $(C', R')$ is statistically indistinguishable from the ideal world of $(C', \bar R)$.
  
  We now construct the extractor for $(C', R')$ and thus define the ideal world of $(C', R')$.
  Without loss of generality, we assume the (non-interactive) opening phase of $(C, R)$, the receiver simply performs a three-outcome general measurement via the canonical square root .
  The extractor for $(C', R')$ is simply to apply extractor on the receiver's state across all the commitments (opened or unopened), and then compute the extracted bit from the outputs of the extractors by \Cref{lem:binding-generalized}.
  We can see that the only difference between the ideal worlds of $(C', R')$ and $(C', \bar R)$ is that for opened commitments, the extractor is applied before the opening phase in $R'$, whereas it is applied after the opening phase in $\bar R$.
  By \Cref{def:stat:binding}, we establish that these two ideal worlds are statistically indistinguishable\footnote{To see this, let the POVM of extractor be $\{\mathcal E_0, \mathcal E_1, \mathcal E_\bot\}$ and let the receiver's POVM be $\{R_0, R_1, R_\bot\}$, then \Cref{def:stat:binding} implies that $R_{1 - b} \preccurlyeq I \otimes (\mathcal E_b + \varepsilon I)$ for $b = 0, 1$ and some negligible quantity $\varepsilon$. Therefore, these two POVMs almost commute when they are implemented by the canonical square root.}, and therefore proving the lemma.
\end{proof}

\begin{definition}
  Let $(C, R)$ be a commitment scheme.
  A commitment scheme $(C', R')$ is called receiver-black-box compiled from $(C, R)$, if its commit phase satisfies the following template:
  \begin{itemize}
    \item $C'$ samples some random bits $r$ and sends a quantum message to $R'$ depending on $r$.
    \item $R'$ commits to a number of bits that is computed from the first message using $(C, R)$.
    \item $C'$ chooses to randomly open some of the commitments, and aborts if $C'$ thinks $R'$ is malicious.
    \item $C'$ computes a (randomized) function that only depends on $r$ and the committed bit $b$, and sends the output to $R'$.
  \end{itemize}
  The function mapping $(C, R)$ to $(C', R')$ is called the receiver-black-box compiler.
\end{definition}
We can also verify that the compiler~\cite[Fig. 3]{BartusekCKM21a} in the second step indeed satisfies this definition by staring at the construction.
Since the statistical binding property is only invoked here to prove the compiled scheme also satisfies statistical hiding, it suffices to prove the following lemma.

\begin{lemma}
  If receiver-black-box compiler satisfies that the compiled scheme satisfies statistical hiding for any $(C, R)$ that satisfies \Cref{def:bckm-binding}, then it also satisfies that the compiled scheme satisfies statistical hiding for any $(C, R)$ that satisfies \Cref{def:stat:binding}.
\end{lemma}
\begin{proof}
  Let $(C, R)$ be a commitment scheme that satisfies \Cref{def:stat:binding}, and let $(C', R')$ be the compiled scheme.
  We need to prove that for any malicious receiver $R'$, his view when the committed bit is 0 is statistically close to that when the committed bit is 1.
  We start with the view where the committer commits to 0.
  Let $(C, \tilde R)$ be the commitment scheme with an inefficient receiver corresponding to \Cref{lem:binding-partial-converse}.
  Similar to the proof before, we first change each invocation of $(C, R)$ to $(C, \tilde R)$.
  Each change is statistically indistinguishable to the malicious receiver $R'$ by statistical indistinguishability of $R$ and $\tilde R$.
  Denote the scheme after the change $(\bar C, R')$ (since we only make chagnes to the honest committer).
  
  Since $(C, \tilde R)$ satisfies \Cref{def:bckm-binding}, by the premise of the problem, the scheme after the change $(C', \bar R)$ satisfies statistical hiding.
  Therefore, we can change the commitment bit from 0 to 1, and the change would be statistically indistinguishable by the statistical hiding property.
  We conclude the proof by undoing the changes one by one in the previous paragraph, and we arrive at the view where the committer commits to 1.
\end{proof}

Since the rest of the proof of \Cref{thm:mpc:statcomm} does not depend on the definition of statistical binding and thus we recover \Cref{thm:mpc:statcomm} with our statistical binding property.